%% file: main.tex
\documentclass[lettersize,journal]{IEEEtran}
\usepackage{amsmath,amssymb,amsfonts}
\usepackage{algorithmic}
\usepackage[ruled,vlined,linesnumbered]{algorithm2e}
\usepackage{graphicx}
\usepackage{subcaption}
\usepackage{textcomp}
\usepackage{xcolor}
\usepackage{booktabs}
\usepackage{multirow}
\usepackage{colortbl}
\usepackage{makecell}
\usepackage{amsthm}
\usepackage{pifont} 

\usepackage{textcomp}
\usepackage{stfloats}
\usepackage{verbatim}
\usepackage{orcidlink}

\definecolor{my_c1}{HTML}{eff8e9}
\definecolor{my_c2}{HTML}{eaf9f8}

\usepackage{etoolbox}        

\newtheorem{theorem}{Theorem}
\newcommand{\model}[1]{Floe}

\newcommand{\GetAvailableMemory}{\text{GetAvailableMemory}}
\newcommand{\PredictMemory}{\text{PredictMemory}}
\newcommand{\PredictLatency}{\text{PredictLatency}}
\newcommand{\SELECTRANK}{\text{SELECT\_RANK}}

\def\BibTeX{{\rm B\kern-.05em{\sc i\kern-.025em b}\kern-.08em
    T\kern-.1667em\lower.7ex\hbox{E}\kern-.125emX}}
\begin{document}

\title{\model~: Federated Specialization for Real-Time LLM–SLM Inference}

\author{Chunlin Tian$^{\orcidlink{0009-0009-5220-1609}}$,
Kahou Tam$^{\orcidlink{0000-0001-5816-6837}}$, 
Yebo Wu$^{\orcidlink{0000-0002-2422-1356}}$, 
Shuaihang Zhong, Li Li$^{\orcidlink{0000-0002-2044-8289}}$,
Nicholas D. Lane $^{\orcidlink{0000-0002-2728-8273}}$,
\\
ChengZhong Xu$^{\orcidlink{0000-0001-9480-0356}}$, ~\IEEEmembership{Fellow,~IEEE}}

\maketitle

\input{sec/01_abs}
\input{sec/02_intro}

\input{sec/03_rw}

\input{sec/04_obs}

\input{sec/06_exp}

\input{sec/07_con}

\bibliographystyle{ieeetr}
\bibliography{main}

\end{document}

%% file: sec/01_abs.tex
\begin{abstract} 
    Deploying large language models (LLMs) in real-time systems remains challenging due to their substantial computational demands and privacy concerns.
    We propose \model~, a hybrid federated learning framework designed for latency-sensitive, resource-constrained environments. \model~ combines a cloud-based black-box LLM with lightweight small language models (SLMs) on edge devices to enable low-latency, privacy-preserving inference. Personal data and fine-tuning remain on-device, while the cloud LLM contributes general knowledge without exposing proprietary weights. A heterogeneity-aware LoRA adaptation strategy enables efficient edge deployment across diverse hardware, and a logit-level fusion mechanism enables real-time coordination between edge and cloud models. 
    Extensive experiments demonstrate that \model~ enhances user privacy and personalization. Moreover, it significantly improves model performance and reduces inference latency on edge devices under real-time constraints compared with baseline approaches.
\end{abstract}


%% file: sec/02_intro.tex
\section{Introduction}

\begin{figure*}[!t]
    \centering
    \includegraphics[width = 0.9
    \linewidth]{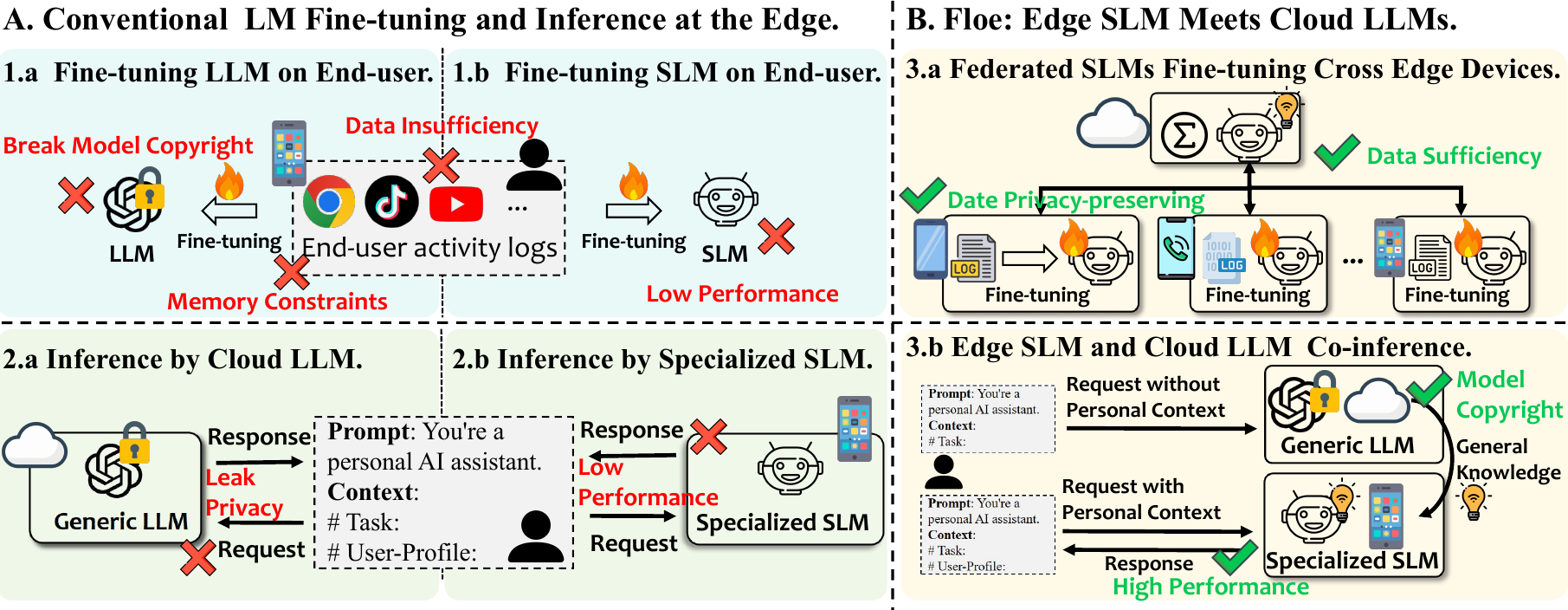}
    \caption{Comparison of existing LM fine-tuning and inference approaches with \model~. \textit{Convention:} (1.a) End-user fine-tunes LLM directly at the edge, breaking model copyright and facing memory constraints. (1.b) End-user fine-tunes specialized SLM at the edge is personalized but underperforms. (2.a) Inference by cloud LLM outperforms but risks privacy. (2.b) Inference by specialized edge SLM is privacy-friendly but underperforms. \model~: (3.a) Federated SLMs fine-tuning cross-edge devices with data sufficiency and privacy. (3.b) Harmonizing LLMs and SLMs improves privacy and performance.}
    \label{fig:intro}
\end{figure*}

\IEEEPARstart{L}{arge} Language Models (LLMs) \cite{GPT4} have become indispensable for natural language processing tasks \cite{microsoft_bing_new_features,BloombergGPT}, yet their reliance on centralized cloud infrastructure creates critical challenges in data privacy, latency, and personalized adaptation (Figure \ref{fig:intro}.1a-\ref{fig:intro}.2a). Federated Language Models \cite{kuang2024federatedscope,wu2025survey} address privacy concerns by decentralizing training while maintaining data locality. However, vanilla FedLLM faces three real-world limitations.
First, proprietary restrictions hinder practical deployment: state-of-the-art LLMs like Gemini Ultra \cite{gemma} ($\$191$M training cost) and GPT-4 \cite{GPT4} ($\$78$M cost) \cite{AI_Charts} remain closed-source due to their immense development investments, with closed-source models consistently outperforming open-source alternatives \cite{AI_Charts}. Second, hardware constraints limit feasibility: even open-source LLMs (e.g., DeepSeek \cite{deepseekai2024deepseekv3technicalreport}, LLaMA \cite{Llama}) require prohibitive resources, deploying LLaMA-65B in FP16 precision requires approximately 130 GB of memory, while full-parameter fine-tuning (FFT) exceeds 800 GB \cite{lv2023full}. Third, edge heterogeneity complicates training: divergent device capabilities, memory limits, and runtime variance create system-level bottlenecks.

To address these challenges, federated language model protocols \cite{wu2025survey} have been proposed, yet they remain fragmented in tackling the full scope of constraints. Existing methods focus on isolated aspects with notable limitations: 1) Parameter-efficient fine-tuning (PEFT) approaches \cite{FederatedScope-LLM,Fate-llm} reduce training overhead but require edge devices to store full LLM copies (e.g., billions of parameters), incurring prohibitive memory costs. 2) Compression-based solutions \cite{offsitetuning, zeroquant,cicc} lower resource demands but risk domain knowledge loss due to lossy tuning and require computationally intensive preprocessing \cite{ma2023llm}.
Moreover, most PEFT-based and compression-based approaches require access to intermediate layer parameters, undermining LLM intellectual property protection.
3) Cryptography-enhanced frameworks \cite{FedTracker,NSDI_FL,NSDI_FL_2,fedipr} protect model IP through ownership verification but introduce communication overhead and fail to guarantee edge device security against attacks. Although small language models (SLMs) \cite{gemma,team2023gemini} now enable on-device deployment, they lack the advanced capabilities and tool integration (Figure \ref{fig:intro}.1b-\ref{fig:intro}.2b) inherent to LLMs, which rely on massive training data and computational power. This situation presents a research question: \textit{How can we balance performance and privacy through collaborative yet privacy-preserving integration of LLMs and SLMs?}

Designing a high-performance federated language models framework for personalized edge-device customization faces three critical challenges: 1) Privacy-preserving Collaboration: Protecting data and model ownership requires keeping raw user data (edge) and LLM parameters (cloud) private. The challenge lies in harnessing decentralized edge data and compute resources to enhance downstream performance without accessing local data or LLM internals. 2) Model Synergy: LLMs offer strong generalization but are hindered by their massive parameter counts on resource-limited edge devices. Meanwhile, SLMs support on-device inference but lack LLMs’ advanced reasoning and tool integration capabilities. Another key challenge is designing a collaborative framework that merges LLM generalization with SLM task-specific specialization for scalable performance. 3) Heterogeneity-aware Orchestration: Edge environments exhibit system heterogeneity, data heterogeneity, and runtime variance. The third challenge is to orchestrate collaborative training that balances accuracy, efficiency, and robustness under these dynamic constraints.

In this paper, we propose \model~, an efficient federated language model framework that harmonizes a central cloud-hosted black-box LLM with edge-deployed lightweight SLMs to enhance downstream task performance while ensuring data privacy, LLM intellectual property, and managing edge device resource constraints. As shown in Figure \ref{fig:intro}.3, \model~ operates in two phases: fine-tuning and inference. During fine-tuning, SLMs serve as the global base model, augmented by Low-Rank Adaptation (LoRA) \cite{lora,tian2024hydralora} adapters for collaborative parameter aggregation to diversify training data and expand device participation. To address hardware and runtime heterogeneity, \model~ introduces heterogeneity-aware adaptive-rank LoRA adapters, which dynamically adjust to varying resource budgets. Additionally, a parameter-free Mixture-of-Experts (MoE) router \cite{moe} mitigates data heterogeneity by intelligently routing tasks to specialized experts. In the inference phase, personalized SLMs on resource-constrained devices leverage local user profiles and activity logs, while the cloud LLM provides high-level knowledge (e.g., strategic planning, outlines, and ``dark knowledge''). The SLMs then integrate this contextual guidance to generate tailored outputs. Extensive experiments demonstrate the effectiveness and practical feasibility of \model~. To the best of our knowledge, \model~ is the \textit{first} federated language model framework across edge devices while logically protecting LLM intellectual property and user privacy through collaborative LLM-SLM generation. Overall, we highlight our contributions as follows:
\begin{itemize}
    \item We propose \model~, a federated language model framework that bridges collaboration between cloud-hosted LLMs and edge-deployed SLMs, jointly addressing intellectual property protection, data privacy, and performance enhancement.
    \item To tackle data/system heterogeneity and runtime variance in edge environments, we design heterogeneity-aware adaptive-rank LoRA adapters and a parameter-free MoE router, balancing model accuracy and training efficiency under resource constraints.
    \item We evaluate \model~ through extensive experiments on real-world edge devices, demonstrating superior accuracy and practical feasibility.
\end{itemize}

%% file: sec/03_rw.tex
\section{Related work}


\subsection{Large Language Models}
\label{Section:logit}
\subsubsection{Training Process}
Large Language Models (LLMs) \cite{GPT4, Llama, gemma} achieve remarkable performance through scaling, typically undergoing two-stage training: (1) Auto-regressive pre-training on broad public datasets (e.g., Wikipedia, web texts) to acquire general knowledge, and (2) supervised fine-tuning on domain-specific data (e.g., medical \cite{chatdoctor}, legal \cite{ChatLaw}, financial \cite{BloombergGPT}) for downstream task specialization. However, public datasets suffer from quality limitations, require extensive curation, and risk depletion by 2026 \cite{2026}, prompting reliance on privately-held data for LLM advancement. While individual parties often possess limited data, collaborative aggregation across distributed private datasets could unlock unprecedented LLM capabilities.

\begin{figure*}[!t]
    \centering
    \includegraphics[width=0.7\linewidth]{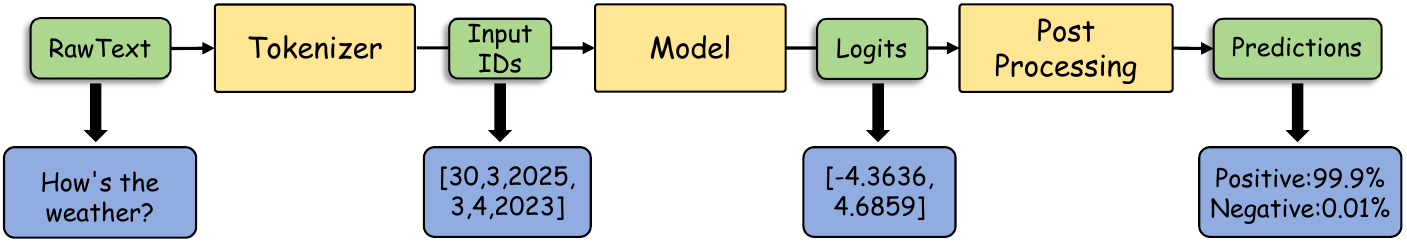}
    \caption{LLMs inference process. The logits play a crucial role in determining the final output.}
    \label{fig:llm_infer}
\end{figure*}

\subsubsection{Inference Process} The inference pipeline of LLMs follows a structured process from input to output generation \cite{logits,logits_1,clone,10682928}, as illustrated in Figure \ref{fig:llm_infer}: (1) Tokenization : Raw text is encoded into token IDs. (2) Model Inference : The model iteratively predicts next-token probability distributions. (3) Decoding : Logits are processed via greedy selection, sampling \cite{Llama}, or advanced strategies to balance diversity and quality. (4) Termination : Generation stops upon reaching a token limit or end token. While closed-source LLMs restrict parameter access, their APIs (e.g., OpenAI API \cite{logits_1}) expose logits for downstream use. Recent works \cite{logit_1,logits,logit_3} demonstrate that ensembling logits from multiple models improves performance. For instance, Contrastive Decoding \cite{logit_6} enhances outputs by selecting tokens that maximize log-likelihood gaps between expert and amateur models. This raises a critical question: \textit{Can we fine-tune LLMs during decoding by modifying logits without accessing black-box model parameters?}

\subsection{Federated Fine-tuning of LLMs}
Federated Learning (FL) \cite{fl} enables privacy-preserving collaborative model training across decentralized devices. Its standard workflow comprises four phases: (1) global model broadcasting from server to clients, (2) local client-side training, (3) model update aggregation from clients to server, and (4) server-side global model updating. To address data scarcity and enhance LLM generalization in downstream tasks, recent studies \cite{FederatedScope-LLM, Fate-llm,tian2024ranking,ning2024fedgcs} have explored FL-based LLM fine-tuning, leveraging distributed private data while maintaining data locality. 
However, existing solutions tend to decouple the issue, concentrating only on specific aspects. Effectively fine-tuning a black-box LLM across multiple devices in real-world scenarios should satisfy the following design requirements.

\begin{table}[!t]
    \centering
    \caption{Popular ML hardware specifications.}
    \resizebox{0.95\linewidth}{!}{
    \begin{tabular}{c|cccc}
    \bottomrule[1.5pt]
    \rowcolor{my_c1}  \textbf{GPU Types}   &  \textbf{Peak Perf.} &\textbf{Memory} & \textbf{Bandwidth} & \textbf{Peak Power}\\
    \toprule[0.75pt]
    \multicolumn{5}{c}{Server-level} \\ \hline
    NVIDIA  A100   & 312 TFLOPS &80GB & 1935 GB/s & 300W\\
    NVIDIA  A40   & 149.7 TFLOPS & 48 GB  & 696 GB/s & 300W\\ \hline
    \multicolumn{5}{c}{Edge-level} \\ \hline
    Jetson Orin NX & 100 TOPS & 16GB & 102.4GB/s&25W\\
    Jetson Orin Nano & 40 TOPS &8GB &68 GB/s &  15W\\
    \toprule[1.5pt]
    \end{tabular}}
    \label{tab:device}
\end{table}


\subsubsection{Memory constraints and model ownership} \label{sec:model}
Full parameter fine-tuning (FFT) of LLMs imposes prohibitive resource requirements, exemplified by LLaMA-65B’s 800GB memory demand for training \cite{Llama, lv2023full,wu2025breaking,wu2024neulite}. As shown in Table~\ref{tab:device}, edge devices face severe computational constraints compared to server-grade GPUs, rendering FFT impractical in decentralized settings. To address this challenge, parameter-efficient fine-tuning (PEFT) methods have emerged, which freeze pre-trained LLM weights and adapt only lightweight modules during federated training. Representative approaches include prompt-based tuning (e.g., PromptFL \cite{FedPrompt}, FedSP\cite{FedSP}) and LoRA (e.g., FederatedScope-LLM \cite{FederatedScope-LLM}, FedIT~\cite{FedIT}). While these techniques reduce computational overhead and communication costs, deploying LLMs \textit{during} training remains memory-intensive for edge devices. For instance, even in FP16 precision, LLaMA-65B requires 130GB for model weights, and LLaMA-7B consumes 14GB, exceeding the capacity of most edge hardware. Second, these approaches inherently require ``white-box'' access to the full model parameters to apply the LoRA adapters, making them incompatible with the proprietary, black-box LLMs that dominate the state-of-the-art.

Compression-based \cite{offsitetuning,zeroquant,tam2024fedhybrid,wu2024heterogeneity} and model-splitting \cite{ceballos2020splitnn,hu2025tightllm,tian2024breaking} approaches reduce resource demands but require white-box access to model parameters or transfer intermediate activations, thereby compromising intellectual property protection and introducing additional privacy risks.

Furthermore, these methods require access to intermediate layer parameters, failing to protect LLM intellectual property. Cryptographic solutions like FedIPR \cite{fedipr} propose watermarking-based ownership verification, but introduce communication overheads and cannot guarantee edge-device security against attacks.
As shown in Table \ref{tab:SLM}, while SLMs enable edge deployment with superior performance in specialized domains (e.g., AceMath-7B-Instruct \cite{liu2024acemath} achieves 63.37\% on MATH\cite{hendrycksmath2021}), they underperform in general tasks (BBH: 29.99\%, MMLU: 26.48\%) due to limited capacity and training data. In contrast, LLMs retain superior general knowledge but cannot be locally deployed on memory-constrained devices. \textit{This dichotomy motivates \model~, which harmonizes both architectures for lightweight, privacy-preserving, and high-performance federated language models.}

\begin{table}[!t]
    \centering
    \caption{Performance comparison of LLMs and SLMs under different benchmarks \cite{open-llm-leaderboard-v2}.}
    \resizebox{0.85\linewidth}{!}{
    \begin{tabular}{c|ccc}
    \bottomrule[1.5pt]
     \rowcolor{my_c1}   \textbf{Model} & \textbf{BBH} (\%) & \textbf{MMLU} (\%) & \textbf{MATH} (\%) \\ \toprule[0.75pt]
    AceMath-7B\cite{liu2024acemath}  & 29.99 &26.48 & \textbf{63.37} \\ \hline
    Qwen2-72B \cite{qwen2} & \textbf{57.48}&\textbf{48.92}&41.77 \\
Llama-3.1-70B\cite{grattafiori2024llama} & 55.93 &47.88 & 38.07\\     \toprule[1.5pt]
    \end{tabular}}
    \label{tab:SLM}
\end{table}

\begin{figure}[!t]
    \centering
    \includegraphics[width=0.85\linewidth]{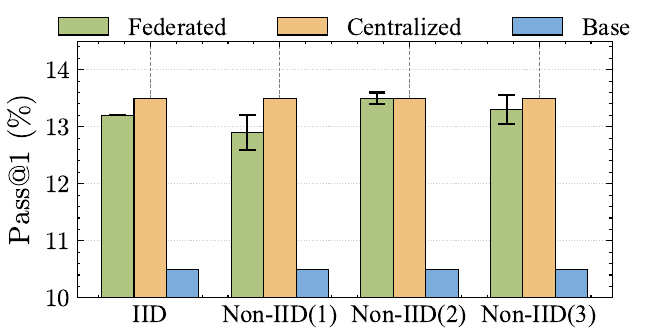}
    \caption{Model performance on in-domain dataset across various Non-IID levels. (Model: Llama-7B; Dataset: CodeAlpaca; Benchmark: Humaneval)}
    \label{fig:Data Heterogeneity}
\end{figure}

\begin{figure}[!t]
    \centering
    \includegraphics[width=0.85\linewidth]{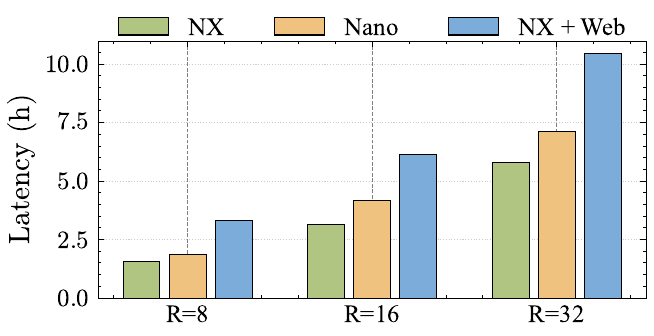}
    \caption{Latency during fine-tuning with different LoRA rank on different edge devices. (Tiny-LLaMA on GSM-8K)}
    \label{fig:rank_latency}
\end{figure}

\subsubsection{Data heterogeneity} \label{sec:data_heter} presents a significant challenge in FL \cite{non-iid,harmony}, which hinders statistical efficiency. For traditional discriminative tasks, data heterogeneity usually refers to the in-domain (ID) data distribution across different clients. To address these issues, several strategies have been proposed in traditional FL, FedProx \cite{fedprox} enhances the stability of the FL process through the addition of a proximal term. Favor \cite{Favor} employs an estimation of data distribution based on discrepancies in uploaded model weights, utilizing statistical utility (non-IID level) as a criterion for selecting a high-quality subset of devices. Recent studies \cite{dataheterogeneity} indicate that architectures based on self-attention, such as Transformers, exhibit enhanced robustness to distribution shifts, thereby improving FL performance. As illustrated in Figure \ref{fig:Data Heterogeneity}, we explore the fine-tuning of Llama-7B using CodeAlpaca \cite{FederatedScope-LLM,codealpaca} across various Non-IID levels, evaluated via the Humaneval \cite{humaneval} benchmark. We employ a Latent Dirichlet Allocation (LDA) based partition strategy to simulate data heterogeneity. The levels Non-IID (1,2,3) correspond to Dirichlet concentration parameters of $\alpha=0.5, 0.3,$ and $0.1$ respectively,
representing increasing degrees of distribution skewness derived from the LDA topic inference. This evaluation demonstrates a high robustness to in-domain heterogeneity.
Despite robustness within a single domain, enhancing the generative capabilities of LLMs in practical applications necessitates a large and diverse dataset that spans multiple domains.
To quantify cross-domain heterogeneity, we use a sentence-embedding model, $\Gamma$ (here adopting BGE \cite{bge-m3}) to evaluate the task embeddings on the multi-takes dataset Flan~\cite{flanv2}. Figure \ref{fig:Lora_embedding} presents a heatmap that visualizes the similarity of task embeddings across different tasks. The results indicate that embeddings within the same domain exhibit higher similarity, highlighting the presence of data heterogeneity across different tasks. \textit{This underscores the need for \model~ to reconcile two priorities: (1) leveraging cross-domain commonalities through federated learning, and (2) accommodating device-specific characteristics via personalization, thereby mitigating data heterogeneity’s impact on model performance and efficiency.}

\begin{figure*}[!ht]
    \centering
    \begin{minipage}[t]{0.3\linewidth}
        \centering
        \includegraphics[width=\linewidth]{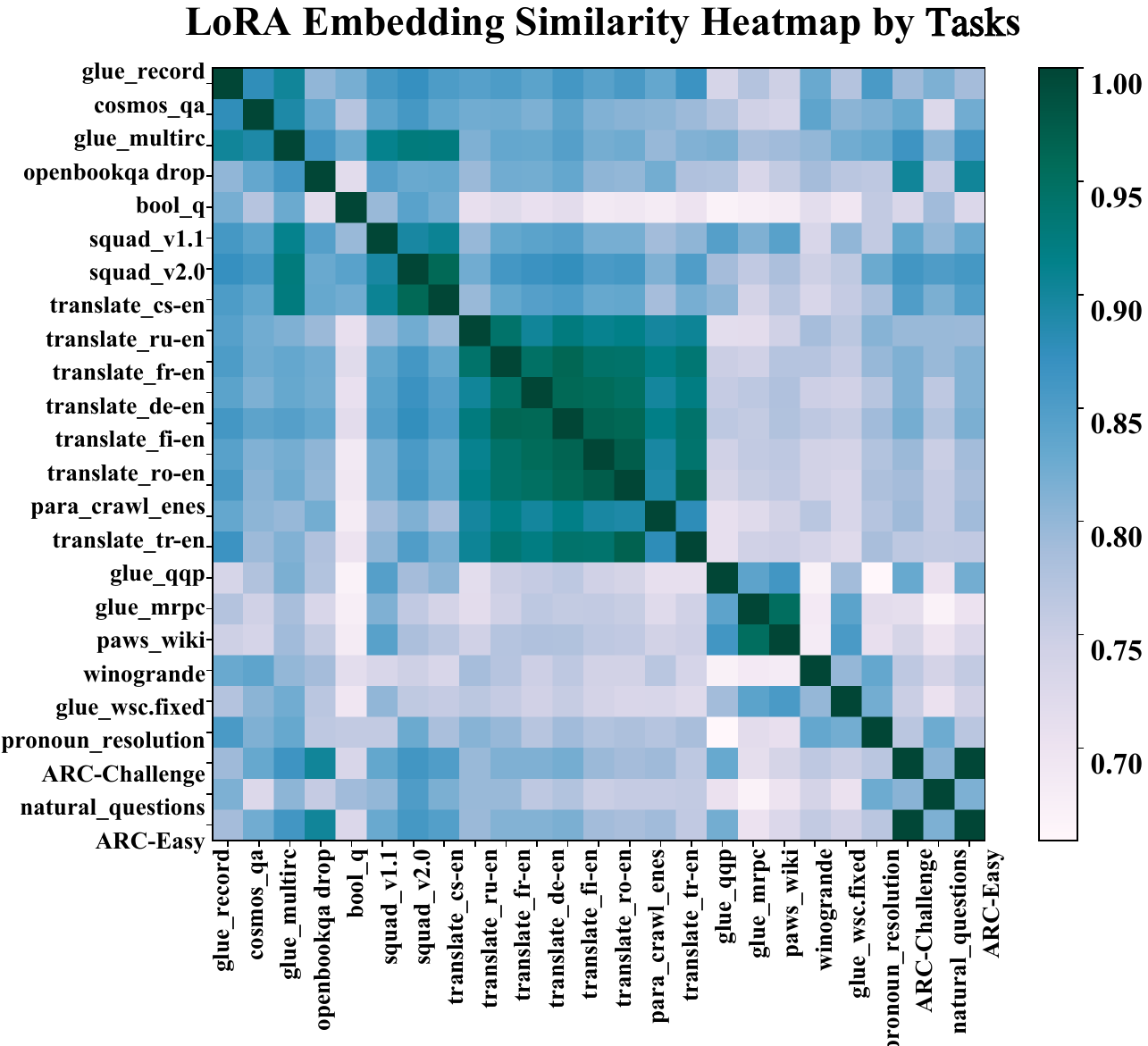}
        \caption{Task embedding heatmap.}
        \label{fig:Lora_embedding}
    \end{minipage}
    \hfill
    \begin{minipage}[t]{0.65\linewidth}
        \centering
    \includegraphics[width=\linewidth]{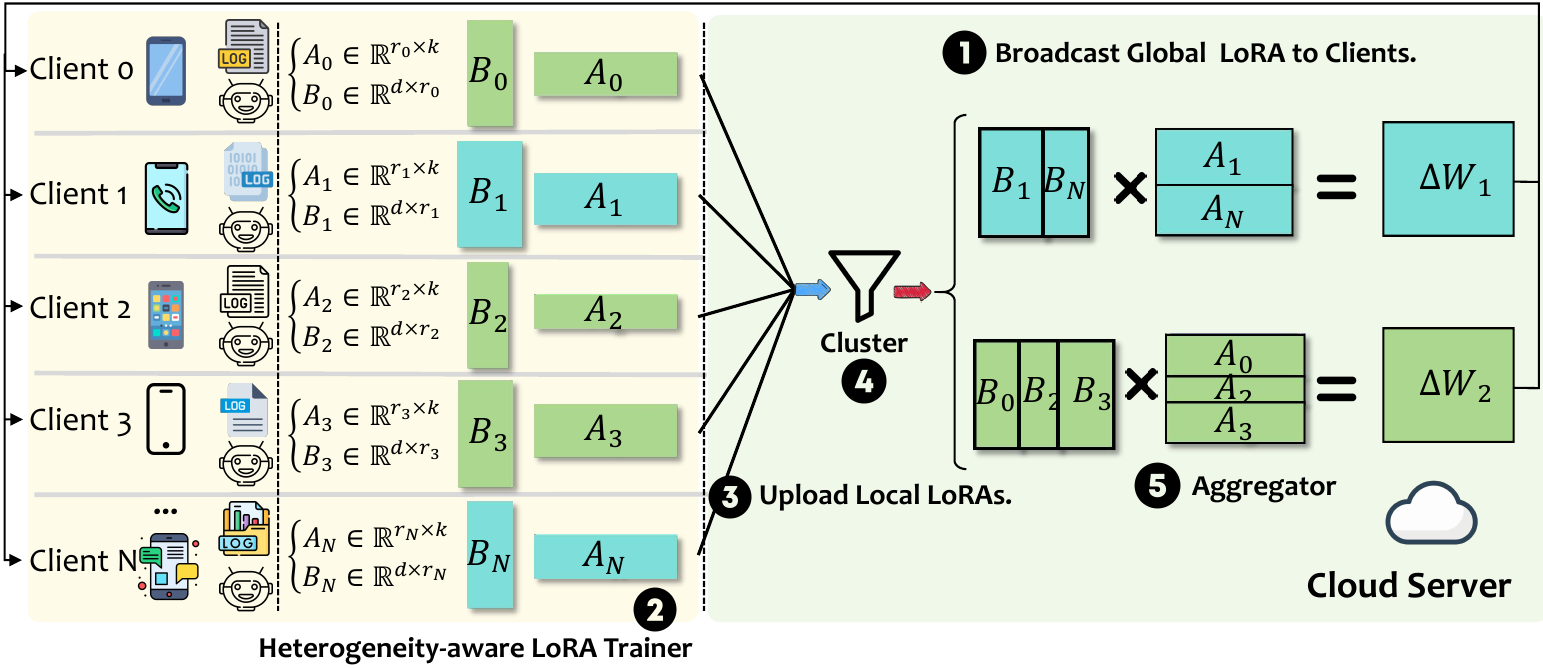}
        \caption{Workflow of \model~ fine-tuning process.}
        \label{fig:overview}
    \end{minipage}
\end{figure*}

\subsubsection{System heterogeneity} arising from performance and efficiency variations across diverse hardware platforms, is especially pronounced in edge environments with thousands of distinct System-on-Chips that exhibit heterogeneous compute and storage capabilities.
Profiling results (Figure \ref{fig:rank_latency}) reveal three observations: (1) hardware disparities significantly affect latency; (2) foreground workloads exacerbate runtime variance; and (3) LoRA rank selection introduces non-trivial computational trade-offs. These findings motivate our heterogeneity-aware rank adaptation.
Traditional approaches \cite{Oort, harmony, FL_selection} address system heterogeneity through device selection strategies but fail to reduce per-device training overhead. For LLMs, high memory demands further restrict eligible device pools. Recent methods \cite{ZeroFL, FedMef} propose dynamic model compression: ZeroFL \cite{ZeroFL} partitions weights into active/inactive subsets for inference and sparse updates during backpropagation, while FedMef \cite{FedMef} preserves post-pruning accuracy via budget-aware parameter extrusion. However, as discussed in Section \ref{sec:model}, compression introduces additional computational overhead, degrades domain-specific knowledge, and undermines privacy through intermediate parameter exposure. In contrast, SLMs mitigate these challenges by streamlining architectures and training data to focus on core task requirements. \textit{This motivates \model~ to develop an adaptive heterogeneity-aware mechanism that dynamically accommodates system heterogeneity and runtime variance while maintaining model performance and privacy.}

%% file: sec/04_obs.tex
\section{\model~: Fine-tuning Phase}

\label{Sect_3}
We now discuss how we leverage the derived principles to design a high-performance federated language model fine-tuning framework. Specifically, we first introduce the architecture and overall workflow of \model~. After that, we discuss each component in detail.

\subsection{System Overview}
Figure \ref{fig:overview} depicts \model~’s federated fine-tuning architecture, which adopts a client-server design with specialized small language models (SLMs) deployed on edge devices (data owners). The process comprises four key stages: \ding{202} To minimize tuning overhead and communication costs, LoRA modules are used as the collaborative fine-tuning component. \ding{203} Heterogeneity Adaptation: Each device freezes its SLM parameters and dynamically adjusts its LoRA trainer within local memory/latency constraints to accommodate hardware and runtime variance. \ding{204} Model Upload: After local training, devices upload their personalized LoRA modules (encoding task-specific adaptations) to the central server. \ding{205} Hierarchical Aggregation: The server clusters LoRA modules based on embedding similarity, and aggregates them via a task-specific aggregator to update the global LoRA module.

\subsection{Heterogeneity-aware LoRA Trainer}
\label{Pruner}
The participating devices in federated fine-tuning often exhibit heterogeneous hardware configurations (e.g., varying compute capabilities and memory budgets) and runtime variances. 
Consequently, a uniform configuration cannot efficiently accommodate heterogeneous devices. While 1B-3B SLMs are classified as 'small' relative to cloud LLMs, performing FFT on them remains infeasible for edge hardware due to excessive memory overheads. Aligning with LoRA's resource-dependent design philosophy, we employ a heterogeneity-aware LoRA trainer to bridge this gap, ensuring feasibility and maximizing participation across constrained edge devices.
In vanilla LoRA, low-rank matrices A (initialized with Kaiming Uniform \cite{he2015delving}) and B (zero-initialized) are sequentially applied to each layer to adapt residual weights. The forward computation is defined as:
\begin{equation}
    y\prime=y + \Delta y = Wx + \Delta Wx =Wx + BAx,   
    \label{eq:lora}
\end{equation}
where $y \in  \mathbb{R}~\textsuperscript{d}$ is the output and $x\in  \mathbb{R}~\textsuperscript{k}$ denote the input. $B \in  \mathbb{R}~\textsuperscript{d×r}, A \in  \mathbb{R}~\textsuperscript{r×k}$ with $r \ll min(d, k)$. $W$ represents the frozen model parameter. To address device heterogeneity, we dynamically configure LoRA ranks (\( r \)) for each device. Given \( N \) devices, we define a set of LoRA modules \( \Phi = \{\phi_1, \phi_2, \dots, \phi_N\} \), where each \( \phi_i \) is trained on local data. For device \( i \), the LoRA rank \( r_i \) is determined by:  
\begin{equation}  
    \left\{  
    \begin{array}{l}  
        A_i \in \mathbb{R}^{r_i \times k} \\  
        B_i \in \mathbb{R}^{d \times r_i},  
    \end{array}  
    \right.  
\end{equation}  
subject to two constraints: 1) Memory Budget: \( M_{r,i} \leq \mathcal{B}_i \) (device \( i \)'s memory capacity). 2) Latency Bound: Predicted training latency \( t_{r,i} \leq \mathcal{T} \) (global round deadline to mitigate straggler effects \cite{Oort,harmony}). The function \textit{PredictLatency(r)} is implemented as a look-up table (LUT) based on the offline profiling data. This allows the trainer to directly retrieve the pre-measured training duration for the specific device and rank, ensuring precise deadline compliance.
To intuitively illustrate the dynamic rank adaptation, we present the selection decision boundaries in Figure \ref{fig:rank_decision}. This adaptive design ensures efficient resource utilization while maintaining training progress synchronization across heterogeneous edge devices. It is orthogonal to compression-like methods, enabling further reduction in additional parameters and computation overhead.
To regulate local data leakage, Floe can employ configurable Local Differential Privacy \cite{abadi2016deep,geyer2017differentially,wei2020federated}. Gaussian noise is injected into the aggregated updates to satisfy differential privacy. The perturbed gradient is given by $\tilde{g} = \bar{g} + \mathcal{N}(0, \sigma_{noise}^2 C^2 \mathbf{I})$, where $\mathcal{N}$ denotes the Gaussian distribution. Here, $C$ is the gradient clipping threshold used to bound the sensitivity, $\sigma_{noise}$ is the noise multiplier, and $\mathbf{I}$ is the identity matrix, ensuring independent noise injection across parameters. This formulation follows the standard DP-SGD protocol \cite{abadi2016deep,geyer2017differentially}. Higher noise ($\sigma$) yields stronger privacy but requires budget tuning to mitigate accuracy loss.

\begin{figure}[!t]
    \centering
    \includegraphics[width=\linewidth]{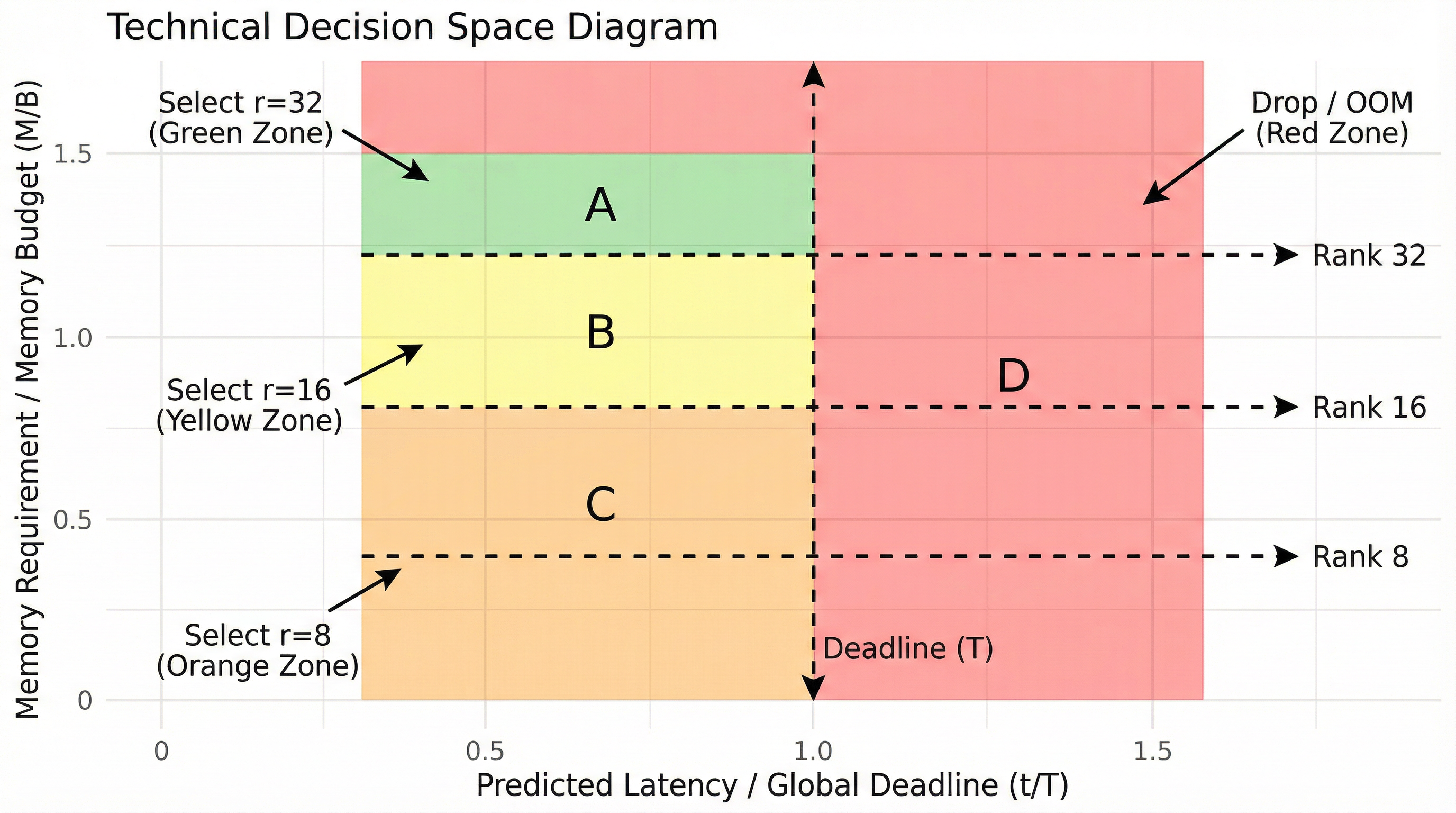}
    \caption{Illustration of Heterogeneity-aware Rank Selection.}
    \label{fig:rank_decision}
\end{figure}

\begin{algorithm}
\caption{Heterogeneity-Aware Rank Selection}
\label{alg:rank_selection}
\begin{algorithmic}[1] 
{\color{black}\STATE \textbf{Input:}
 $\mathcal{R}_{set}$: Available LoRA ranks;
 $\mathcal{T}$: Global latency deadline;
 $\GetAvailableMemory()$: Gets current free memory;
 $\PredictMemory(r)$: Predicts memory use for rank $r$;
$\PredictLatency(r)$: Predicts latency for rank $r$.
\STATE \textbf{Output:} $r_{selected}$: The optimal LoRA rank for the current round.
\STATE \textbf{Function} \SELECTRANK()
    \STATE $r_{selected} \leftarrow None$
    \STATE $\mathcal{B}_i \leftarrow \GetAvailableMemory()$ 
    \FOR{$r$ \textbf{in} $sorted(\mathcal{R}_{set}, \text{reverse}=\text{true})$}
        \STATE $M_{r,i} \leftarrow \PredictMemory(r)$
        \STATE \textbf{/* Stage 1: check memory constraint */}
        \IF{$M_{r,i} \le \mathcal{B}_i$}
            \STATE $t_{r,i} \leftarrow \PredictLatency(r)$
            \STATE \textbf{/* Stage 2: check latency constraint */} 
            \IF{$t_{r,i} \le \mathcal{T}$}
                \STATE $r_{selected} \leftarrow r$
                \RETURN $r_{selected}$ 
            \ENDIF
        \ENDIF
    \ENDFOR    
    \RETURN $r_{selected}$ 
\STATE \textbf{End Function}}
\end{algorithmic}
\end{algorithm}

\subsection{Task-specific Router/Aggregator}
To address cross-task heterogeneity under privacy and resource constraints, we design a task-specific LoRA aggregation strategy that extends beyond naive averaging. Each LoRA module $\phi_i$ is first projected into a task-aware embedding space using a domain-conditioned encoder $E(\phi_i)$, capturing both its adaptation semantics and fine-tuning dynamics. This embedding reflects not only parameter updates but also behavioral traits relevant to downstream generation, which are critical for privacy-sensitive, personalized tasks executed on edge devices. We then compute cosine similarities between embeddings to quantify functional divergence and form semantically coherent clusters:
\begin{equation}
s_{i,j} = \cos\left(E(\phi_i), E(\phi_j)\right).
\end{equation}
To quantify functional divergence and form semantically coherent groups, we adopt a clustering-based aggregation mechanism. We first apply a k-means clustering algorithm to the embeddings $E(\phi)$ of all $N$ client modules. This partitions the modules into $M$ clusters, where each cluster represents a distinct, discovered task. The optimal number of clusters, $M$, is not a fixed setting. It is determined dynamically in each aggregation round by analyzing the cluster quality using the silhouette score. This enables the server to adapt to the actual task distribution among the participating clients. Once the $M$ clusters are formed, we proceed with task-specific aggregation. Considering $M$ distinct downstream tasks, $\mathcal{T} = {T_1, T_2, \ldots, T_M}$, and assuming that $k$ LoRAs correspond to the same task $T_t$ based on their similarity scores, we represent the parameters of each LoRA $\phi_i$ by $\Theta_i$. The collective parameter for task $T_t$ is then computed as the average of these parameters:
\begin{equation}
\Theta_{T_j} = \frac{1}{k} \sum_{i=1}^{k} \Theta_i, \quad \phi_i \in T_t.
\label{eq:param}
\end{equation}
Unlike traditional federated methods that enforce a single shared model across diverse clients, this clustered aggregation preserves per-task specialization while mitigating negative transfer, especially under non-IID conditions and misaligned user intents (see Section~\ref{sec:data_heter}). As visualized in Figure~\ref{fig:overview}, this method produces efficient, task-specific updates well-suited for federated LLM fine-tuning under privacy constraints.

While Floe currently utilizes a deadline-based synchronous protocol, the framework naturally supports asynchronous extension to handle extreme device unavailability. 
Unlike global model aggregation which blocks for all clients, Floe's \textit{Task-Specific Aggregator} allows for \textit{Cluster-wise Asynchronous Updates}. The server can update a specific cluster center $\Theta_{T_j}$ immediately upon receiving a local LoRA $\phi_i$ that maps to task $T_j$, independent of pending updates for other clusters. 
To further mitigate the impact of stragglers uploading outdated parameters (``stale'' updates), a \textit{staleness-aware} weighting mechanism is introduced. Let $\tau_{m}=t_{now}-t_{m}^{start}$ denote the time lag of client $m$. We replace the uniform averaging in Eq. \ref{eq:param} with a time-decayed weighted aggregation:
\begin{equation}
\label{eq:async}
\Theta_{T_{j}}^{async}=\frac{\sum_{m=1}^{k}exp(-\beta\tau_{m})\cdot\Theta_{m}}{\sum_{m=1}^{k}exp(-\beta\tau_{m})}
\end{equation}
where $\beta$ is a decay hyperparameter. This ensures that fresher updates contribute significantly to the global expert pool, while slower devices can still participate with reduced influence, preventing model degradation from obsolete gradients.

\subsection{Convergence Analysis} 
We provide a theoretical guarantee for Floe's training process. Since Floe employs a clustering-based aggregation strategy, we analyze the convergence within a specific task cluster $\mathcal{C}$ with $K$ clients. We model the heterogeneity-aware LoRA update as a compressed gradient estimator.
\subsubsection{Assumptions}
We adopt standard assumptions for non-convex Federated Learning:
\begin{itemize}
    \item \textbf{Assumption 1 (L-Smoothness):} $F_k$ is $L$-smooth: $\|\nabla F_k(\mathbf{w}) - \nabla F_k(\mathbf{v})\| \le L \|\mathbf{w} - \mathbf{v}\|$.
    \item \textbf{Assumption 2 (Bounded Variance):} Stochastic gradients are unbiased with bounded variance $\sigma^2$: $\mathbb{E}[\|\nabla F_k(\mathbf{w}; \xi) - \nabla F_k(\mathbf{w})\|^2] \le \sigma^2$.
    \item \textbf{Assumption 3 (Bounded Dissimilarity):} The non-IID degree is bounded by $\kappa^2$: $\|\nabla F_k(\mathbf{w}) - \nabla F(\mathbf{w})\|^2 \le \kappa^2$.
    \item \textbf{Assumption 4 (Rank Approximation Error):} The adaptive LoRA rank $r_k$ acts as a compression operator $Q_{r_k}(\cdot)$. The approximation error is bounded by $\delta$: $\mathbb{E}[\|\mathbf{g} - Q_{r_k}(\mathbf{g})\|^2] \le (1 - \delta) \|\mathbf{g}\|^2$, where $0 < \delta \le 1$.
\end{itemize}
\subsubsection{Main Result}
Let $\eta$ be the learning rate and $E$ be the local steps. The global model updates via $\mathbf{w}_{t+1} = \mathbf{w}_t - \eta \sum_{k=1}^K p_k \Delta \mathbf{w}_t^k$.
\begin{theorem}
\label{thm:convergence}
(Convergence). Under Assumptions 1-4, if $\eta \le \frac{1}{30LE}$, the average squared gradient norm after $T$ rounds is bounded by:
\begin{equation}
\begin{split}
    \frac{1}{T} \sum_{t=0}^{T-1} \mathbb{E} [\|\nabla F(\mathbf{w}_t)\|^2] & \le \frac{2[F(\mathbf{w}_0) - F^*]}{\eta E T} \\
    & \quad + \mathcal{O}\left( \frac{\sigma^2}{K E} + \eta^2 L^2 E (\kappa^2 + \sigma^2) \right) \\
    & \quad + \mathcal{O}\left( (1-\delta) \sigma^2 \right)
\end{split}
\end{equation}
\end{theorem}
\textit{Proof:}
We start with the $L$-smoothness expansion of the global objective function $F(\mathbf{w}_{t+1})$. By substituting the aggregated LoRA update rule:
\begin{equation}
\begin{split}
    \mathbb{E}[F(\mathbf{w}_{t+1})] & \le \mathbb{E}[F(\mathbf{w}_t)] - \frac{\eta}{2} \mathbb{E}\|\nabla F(\mathbf{w}_t)\|^2 \\
    & \quad + \frac{\eta L}{2} \underbrace{\mathbb{E} \left\| \sum_{k=1}^K p_k (\mathbf{g}_t^k - Q_{r_k}(\mathbf{g}_t^k)) \right\|^2}_{\text{Rank Approx. Error}}
\end{split}
\end{equation}
The error term decomposes into the standard stochastic variance $\sigma^2$, the heterogeneity drift $\kappa^2$, and the LoRA rank approximation error $(1-\delta)$. Floe's \textit{Task-Specific Aggregator} explicitly minimizes $\kappa^2$ by clustering clients with similar gradients, while the heterogeneity-aware trainer trades off the rank error $(1-\delta)$ to ensure system participation. Summing over $T$ rounds yields the bound.
\hfill $\square$

\begin{figure}[!t]
    \centering
        \includegraphics[width=0.95\linewidth]{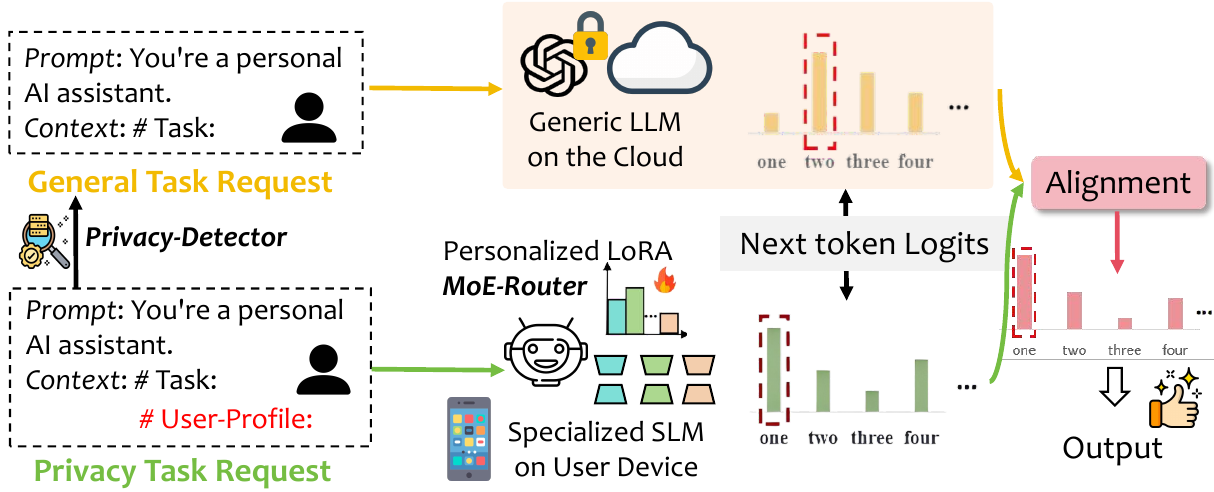}
        \caption{Illustration of \model~ inference process.}
        \label{fig:inference}
\end{figure}

\section{\model~: Inference Phase}
This section describes the \model~ inference pipeline, which seamlessly integrates: Cloud-hosted LLMs leverage their broad generalization capabilities for high-level knowledge, and on-device SLMs specialize in privacy-sensitive user instructions via local context adaptation. To optimize performance in heterogeneous task environments, we propose a parameter-free MoE router that dynamically merges LoRA modules, enabling adaptive task routing without introducing additional trainable parameters.

\subsection{Privacy Detector}
To ensure that sensitive prompts never reach the cloud LLM,  \model~ deploys a lightweight on-device Privacy Detector. Executing ahead of the router in Figure \ref{fig:inference}, it decides—within sub-millisecond latency—whether a request must remain on the local SLM.
The detector is a two-stage heuristic tuned for speed and coverage. Stage 1: applies rule-based screening: regular expressions catch numeric identifiers (phone numbers, credit-card numbers, national IDs) while a compact named-entity list flags tokens tied to health, finance, location, or family relations. the privacy flag is set to true, and the prompt bypasses the cloud.
Prompts that clear Stage 1 enter Stage 2, a semantic back-off. The text is embedded with the BGE-base model \cite{bge-m3}; cosine similarity is then computed against five centroids covering location and user-profile domains. Scores above the tuned threshold $\tau$ likewise trigger the privacy flag, capturing paraphrases and obfuscated wording that rules miss.

\begin{algorithm}
\caption{Privacy\,Detector}
\label{alg:privacy}
\KwIn{Prompt $\mathbf{x}$}
\KwOut{Privacy (\textbf{true} if $\mathbf{x}$ must \emph{not} be sent to the cloud)}
\smallskip
\SetKwProg{Fn}{Function}{:}{}
\Fn{\textsc{Detect}$(\mathbf{x})$}{
    \textbf{/* Stage 1: rule-based filter */} \\
    \If{\,\textsc{RegexMatch}$(\mathbf{x})$ \textbf{or} \textsc{NERMatch}$(\mathbf{x})$}{
        \Return \textbf{true} \tcp*{keep on device}
    }
    \textbf{/* Stage 2: semantic back-off */} \\
    $\boldsymbol{e} \leftarrow \Gamma(\mathbf{x})$ \tcp*{BGE-base embedding}
    \ForEach{domain centroid $\boldsymbol{\mu}_{j}\!\in\!\{\textit{health},\textit{finance},\textit{legal},\textit{location},\textit{profile}\}$}{
        $s_{j} \leftarrow \cos(\boldsymbol{e},\,\boldsymbol{\mu}_{j})$
    }
    \If{$\displaystyle \max_{j} s_{j} > \tau$}{
        \Return \textbf{true}
    }
    \Return \textbf{false}
}
\end{algorithm}

\subsection{Prompt-wise MoE-based Router Plugin}
In practical settings, the assumption that a single, universal solution—such as directly averaging the LoRA modules for all downstream tasks—will suffice is inadequate. This inadequacy arises because end-user requests span a diverse range of prompts, each associated with distinct tasks. Consequently, as shown in Figure \ref{fig:design_moe}, an MoE \cite{moe} router has been developed to dynamically and effectively merge LoRA modules for each mixed-task prompt. This approach significantly extends the plug-and-play capabilities of \model~.
An MoE layer comprises a router network, denoted as $R$, and a set of experts, $\textbf{E} = (E_{1}, E_2, \ldots, E_{N})$, where each expert $E_{i}$ corresponds to a tuned LoRA module. Additionally, we define $\textbf{A}=(A_1, A_2, \ldots, A_N)$ and $\textbf{B}=(B_1, B_2, \ldots, B_N)$ as the sets of matrices, each containing $N$ LoRAs.
For an input $x_{i}$, the output $y_i$ is expressed as:
\begin{align}
    y_i &= \mathbf{W} x_i + \sum_{j=1}^N \omega_{ij} \cdot E_j\left(x_i\right)  \notag \\
    &= \mathbf{W} x_i + \sum_{j=1}^N \omega_{ij} \cdot B_j A_j x_i,
\end{align}
where $\omega_{ij}$ modulates the contribution weights of each expert.
Unlike conventional MoE architectures that rely on trainable gating networks (e.g., GShard \cite{GShard}, Switch Transformer \cite{Switch}), we propose a parameter-free, prompt-driven router tailored for edge deployment. This design eliminates additional training and synchronization overhead, enabling efficient specialization without compromising latency or memory budgets. Given a mixed-task prompt $x$, we compute its embedding $\Gamma(x)$ using a sentence encoder (BGE). To enable this comparison, each aggregated expert LoRA module $\phi$ (from Section III.C) must also be represented in this embedding space. This is achieved via a one-time, pre-calculation on the server. The server computes a static embedding for each expert $\phi$ by averaging the embeddings of $k$ server-held, non-private representative samples ($k_{i, \phi}$) that define that expert's domain (e.g., generic, public medical questions for a ``medical'' expert):
\begin{equation}
\Gamma(\phi) = \frac{1}{k} \sum_{i=1}^{k} \Gamma(k_{i \cdot \phi}).
\end{equation}
This static, pre-computed embedding $\Gamma(\phi)$, which contains no private user data, is then distributed to the edge devices to be used by the router. The router calculates cosine similarity between the prompt and each expert:
\begin{equation}
s(x, \phi) = \cos(\Gamma(x), \Gamma(\phi)),
\end{equation}
and applies a softmax to derive the gating weights:
\begin{equation}
\Omega = \text{softmax}(\mathbf{s}_x).
\end{equation}
This method has three key advantages for real-world deployment: (1) No communication overhead—gate weights need not be learned or synchronized;
(2) Low latency—matching completes in under 1 ms on Jetson-class devices;
(3) Interpretability—each LoRA expert aligns with a known domain (e.g., medical, legal), supporting transparent, domain-specific control.
Moreover, when downstream tasks evolve or new domains emerge, traditional MoE requires expensive retraining of the router. In contrast, our soft MoE enables flexible plug-and-play updates—experts can be added, removed, or replaced without modifying or retraining the core routing mechanism.

\begin{figure}[!t]
        \centering
        \includegraphics[width=0.65\linewidth]{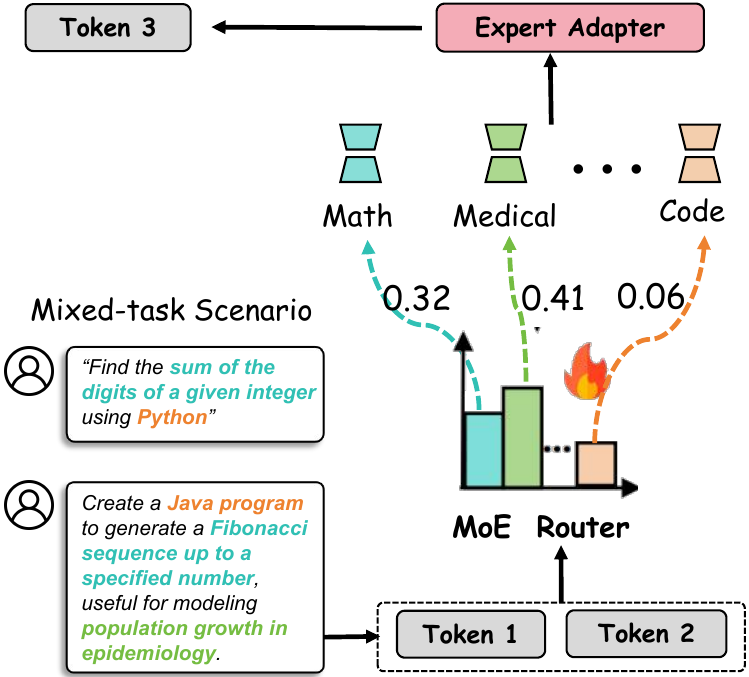}
        \caption{Illustration of MoE inference. The router computes similarity scores for all expert clusters (e.g., Math, Medical, Code). A Softmax function converts these scores into gating weights (e.g., 0.32, 0.41). The final model adaptation is a weighted summation of all experts, allowing for mixed-domain reasoning without hard selection thresholds.}
        \label{fig:design_moe}
\end{figure}

\subsection{Logit-level LLM--SLM Alignment}
\label{Tuner}
To leverage complementary strengths, \model~ performs logit-level fusion between the cloud LLM (general knowledge) and the on-device SLM (personalized context). For each decoding step $i$, both models independently produce probability distributions over the vocabulary.
Let $r_{<i}$ denote the previously generated tokens. The cloud model predicts:
\begin{equation}
P^{\text{LLM}}_i = \theta_l(r_{<i}, t),
\end{equation}
where $t$ denotes non-sensitive task metadata. The edge model predicts:
\begin{equation}
P^{\text{SLM}}_i = \theta_s(r_{<i}, t, C),
\end{equation}
where $C$ represents private user context retained locally.
A lightweight alignment module computes a fusion weight $w \in [0,1]$ based on the divergence between the two distributions:
\begin{equation}
w = \sigma\big(\text{MLP}([P^{\text{SLM}}_i ; P^{\text{LLM}}_i])\big),
\end{equation}
where $\sigma(\cdot)$ denotes the sigmoid activation function.
The final output distribution is given by:
\begin{equation}
P^{\text{out}}_i = 
w \cdot P^{\text{SLM}}_i 
+ (1 - w) \cdot P^{\text{LLM}}_i.
\label{eq:logit_fusion}
\end{equation}
This logit-level interpolation preserves the general reasoning capability of the LLM while incorporating user-specific adaptations from the SLM, without exposing proprietary model parameters or transmitting raw private data.

\subsection{Real-Time Robustness and Fallback Strategy}
\label{sec:robustness}

While the hybrid LLM-SLM architecture leverages global knowledge from the cloud, reliance on remote logits inevitably introduces dependencies on network stability and API availability. To ensure strict real-time guarantees despite these external variances, Floe employs a two-tier mitigation strategy.

First, the Privacy Detector serves as a natural traffic offloader. By processing privacy-sensitive and often latency-critical user instructions (e.g., personal context management) entirely on the edge, Floe insulates the most frequent user interactions from network jitter.
Second, for non-private queries that require cloud collaboration, we introduce a timeout-based fallback mechanism to mitigate tail latency. During the token-wise generation, the on-device SLM and cloud LLM inference occur in parallel. If the cloud logits $\mathcal{P}_{i}^{LLM}$ fail to arrive within a predefined latency budget $\tau$ (e.g., 200ms), the alignment module dynamically forces the fusion weight $w \to 1$ in Eq.~\ref{eq:logit_fusion}. This effectively enables the system to ``fall back'' to the local SLM prediction for the current token, prioritizing fluid user experience over global knowledge integration when network conditions degrade. This mechanism also optimizes cost efficiency by preventing indefinite waiting times and effectively managing API quota consumption.

%% file: sec/06_exp.tex
\section{Experiment}
\subsection{Experiment Setup}
\subsubsection{Infrastructure} We deploy \model~ on a server–client testbed consisting of one NVIDIA A40 GPU (server) and 15 heterogeneous Jetson devices (clients, 10 Jetson Nano and 5 Jetson Orin NX). Background workloads are injected to simulate runtime variance~\cite{microsoft_browser_trace,9799238}. We utilize Codecarbon \cite{codecarbon} for tracking latency and energy consumption during the fine-tuning process.

\subsubsection{Models and Datasets}
To validate \model~'s performance and efficiency, we adopt a hybrid experimental setup combining commercial API-based and open-source models. For closed-source LLMs, we use \textit{GPT-4-Turbo} \cite{GPT4} as the cloud-hosted backbone and fine-tune lightweight edge SLMs (\textit{TinyLlama} \cite{zhang2024tinyllama}) for personalized adaptation. For open-source experiments, we deploy \textit{Gemma-7B}~\cite{gemma} as the server-hosted LLM and \textit{Gemma-2B} as the collaborative SLM. Our evaluation leverages two datasets: (1) \textit{CoGenesis} \cite{logit_4}, a real-world user interaction dataset for privacy-sensitive tasks partitioned by User IDs to reflect real-world data islands, and (2) \textit{Flanv2}~\cite{flanv2}, a multi-task benchmark spanning 10 distinct clusters~\cite{wei2021finetuned} which are distributed across clients to simulate task-level non-IID heterogeneity: Struct-to-Text Conversion, Translation, Commonsense Reasoning, Sentiment Analysis, Closed-Book QA, Paraphrase Detection, Coreference Resolution, Reading Comprehension, Commonsense-Enhanced Reading Comprehension, Natural Language Inference.  

\subsubsection{Evaluation Metrics} 
We assess \model~ through two critical dimensions:  
\textit{(1) Downstream Task Performance}: Using the Big-Bench Hard (BBH)~\cite{bbh} benchmark, which contains 23 challenging multi-domain reasoning tasks from the BIG-Bench suite, we measure the customized model's world knowledge and problem-solving capabilities. Results on CoGenesis are validated using its dedicated test set.  
\textit{(2) System Efficiency}: We quantify latency and energy consumption using the WikiText2~\cite{Wikitext2} dataset, reflecting real-world deployment constraints on edge devices. 

\begin{table*}
    \centering
    \caption{Results for instruction-tuning on Specialized SLM (Gemma-2B), generic
black-box LLM (Gemma-7B) and other fine-tuning methods on the BBH benchmark (5-shot). \model~ outperforms other baselines, and improves performance over the LLM-base by 3.6\%, and 14.2\% improvement over vanilla smaller LMs.}
    \resizebox{0.95\linewidth}{!}{
    \begin{tabular}{c|cccc|ccc|c}
    \toprule[1.5pt]
    \rowcolor{my_c1}  \textbf{Task}    & \textbf{SLM-base} & \begin{tabular}[c]{@{}c@{}}\textbf{SLM-Local}\end{tabular}  & \begin{tabular}[c]{@{}c@{}}\textbf{SLM-FedAvg}\end{tabular} & \begin{tabular}[c]{@{}c@{}}\textbf{SLM-FedProto}\end{tabular} & \textbf{LLM-base}  & \begin{tabular}[c]{@{}c@{}}\textbf{LLM-FedAvg}\end{tabular} & \textbf{LLM-FedMoE} & \begin{tabular}[c]{@{}c@{}}\textbf{\model~} \end{tabular}   \\ \midrule[0.75pt]
         Boolean Expressions     & 76.1 & 73.3 & 73.9 & 74.3 & 82.0 & 80.4 & 82.4 & 82.9  \\
     Causal Judgement    & 47.3 & 47.8 & 47.8 & 47.8 & 56.0 & 56.6 & 57.1  & 58.2\\
     Date Understanding   & 25.1 & 16.6 & 24.3 & 26.9 & 47.3  & 50.6 & 51.4 & 52.7\\
     Disambiguation Qa      & 34.0 & 38.9 & 34.0 &  28.2 & 53.1 & 56.7 & 66.1 & 69.0\\
     Dyck Languages     & 24.3 & 29.6 & 25.9  & 25.7 & 51.8 & 52.7 & 66.1  & 69.4 \\
     Formal Fallacies     & 46.2 & 46.2 & 46.2  & 46.1 & 50.6 & 57.1 & 51.0 & 53.1 \\
     Geometric Shapes     & 29.6 & 29.6 & 29.6  & 30.6 & 34.3 & 36.3 & 37.6 & 41.8\\
     Hyperbaton   & 57.5 & 54.7 & 57.1  & 57.6 & 56.7 & 58.0 & 61.2 & 61.4\\
     LDeduction (5)     & 23.5 & 22.3 & 22.7  & 24.1 & 30.2 & 29.8 & 37.6& 40.8 \\
     LDeduction (7)   & 15.4 & 15.8 & 15.4 & 16.3 &29.0 & 26.1 & 26.9& 35.1 \\
     LDeduction (3)    & 33.2 & 39.3 & 34.8  & 36.3 & 49.0 & 47.3 & 50.6 & 52.7\\
     Movie Recommendation    & 44.1 & 43.3 & 44.1  & 48.2 & 69.8& 70.2 & 73.5 & 75.1\\
     Multistep Arithmetic Two  & 0.4 & 0.4 & 1.2  & 2.9  & 1.6& 2.9 & 2.0 & 2.9\\
     Navigate    & 57.5 & 57.5 & 57.5  & 57.6 & 58.4 & 58.4 & 59.6& 60.6 \\
     Object Counting     & 29.6 & 28.3 & 30.0 & 32.2 & 43.3 & 44.9 & 44.1 & 45.3\\
     Penguins in a Table  & 27.3 & 25.9 & 26.6 &  31.2  & 45.4 & 46.1 & 44.7 & 44.7\\
     Reasoning    & 17.8 & 14.6 & 17.4  & 16.3 & 43.7 & 42.4 & 40.4 & 41.2 \\
     Ruin Names   & 13.8 & 15.0 & 12.6  & 18.4  & 29.0 & 33.1 & 31.8 & 35.1 \\
     Salient & 17.4 & 15.8 & 14.6  & 18.4 & 34.3  & 35.9 & 35.1& 35.5\\
     Snarks & 53.1 & 56.6 & 48.0  & 49.1  & 59.5 & 61.3 & 57.8  & 59.5\\
     Sports Understanding  & 55.5 & 59.5 & 54.7  & 50.2 & 76.7 & 73.5 & 74.3 & 75.9 \\
     Temporal Sequences & 17.4 & 19.4 & 17.4  & 18.0 & 17.6  & 15.5 & 17.1 & 19.2 \\
     Tracking Shuffled Objects (5) & 18.2 & 21.5 & 17.0 & 20.4 & 13.9 & 16.7 & 13.9 & 15.9\\
     Tracking Shuffled Objects (7) & 12.1 & 10.1 & 12.6  & 13.1  & 11.4 & 12.2 & 12.2 & 12.2\\ 
     Tracking Shuffled Objects (3) & 35.2 & 37.2 & 36.8   & 38.0 & 33.9 & 32.7 & 33.1 & 35.5\\ 
     Web Of Lies    & 50.6 & 50.6 & 50.6  & 50.6 & 54.3 & 53.9 & 51.4  & 55.1\\ 
     Word Sorting   & 7.7 & 6.1 & 8.5 & 6.9  & 22.4 & 22.4 & 23.3& 24.6\\ \midrule[0.75pt]
     \rowcolor{my_c2} \textbf{Average} & 32.21 & 32.43 & 31.96 &  \underline{32.82} & 42.79 & 43.81 & \underline{45.12} & \textbf{46.39}\\
     \bottomrule[1.5pt]
    \end{tabular}}
    
    \label{tab:overall}
\end{table*}

\subsubsection{Baselines} We compare \model~ with the following baselines for evaluation purposes.
\textit{(1) SLM-base} ~\cite{gemma}: Directly evaluates the untuned small open-source Gemma-2B on the BBH benchmark.
\textit{(2) SLM-Local}: Local device (data owner) uses a single LoRA module to independently fine-tune the Gemma-2B model using its private dataset, without any coordination, then evaluated on BBH benchmark.
\textit{(3) SLM-FedAvg} ~\cite{fedavg}: 
Using Gemma-2B as the global model, each local device independently trains LoRA modules on its own data. The server then applies the federated averaging method to directly average the updated LoRA parameters from all devices.
\textit{(4) SLM-FedProto} ~\cite{fedproto}: To address data heterogeneity by mitigating the divergence between local and global class prototypes, thereby aligning the locally generated prototypes with their global counterparts.
\textit{(5) LLM-base} ~\cite{gemma}: The untuned Gemma-7B model is directly evaluated on the BBH benchmark.
\textit{(6) LLM-FedAvg}: Using Gemma-2B as the proxy model, and Gemma-7B as the base LLMs. After local tuning, the server (model owner) directly merges all updated LoRA with equal weight. Subsequently, the server employs proxy inference to evaluate the performance on the BBH dataset benchmark.
\textit{(7) LLM-FedMoE} ~\cite{LoRAMoE_1}: 
To address the issue of assigning equal importance to all LoRA parameters, a router network comprising a dense layer with trainable parameters is used. This network evaluates the importance of each LoRA and selects the top-3 most significant LoRA modules for the corresponding inputs.

\subsubsection{Hyperparameter Settings and Reproducibility}
For local training, we train each LoRA module according to Stanford-Alpaca~\cite{alpaca} format on NVIDIA Jetson Nano for a spectrum of tasks with 3 epochs. We set the learning rate to $1e^{-5}$, and dynamically adjust the learning rate using a cosine annealing scheduler. The training batch size and gradient accumulation steps are both 4. We enable training with half-precision floating-point parameters to reduce memory usage and speed up training. During the proxy inference process, for fair comparison, we conduct three runs to determine average performance, each utilizing different 5-shot examples per task.

\subsection{Overall Performance}
\noindent\textit{RQ1: Can \model~ effectively proxy fine-tune closed-source LLMs using smaller models while maintaining excellent performance on multi-task downstream inputs?}

Table \ref{tab:overall} presents the overall comparison results of all schemes evaluated on the multi-task benchmark (BBH). It is evident that \model~ significantly outperforms other baselines across various tasks. Notably, \model~ improves performance accuracy by 14.2\% compared to the standard small model and by 3.6\% over the conventional LLM. Moreover, \model~ proves particularly effective in handling complex mixed-task inputs, outperforming FedAvg by 2.6\% and FedMoE by 1.3\%.
 Unlike traditional models, obtaining even small accuracy improvements for LLMs can require significant training overhead.
\model~'s gain is achieved with significantly lower system cost.
As shown in Table \ref{tab:Co}, the hybrid system (LLM + SLM) outperforms standalone models in most tasks. For instance, combining GPT-4-turbo with an SLM achieves scores of 7.84 in Instructions and 7.65 in Paper Abstracts, surpassing both vanilla GPT-4-turbo (6.46/7.94) and TinyLlama-1.1B (7.66/7.32). Notably, GPT-3.5-turbo paired with an SLM exceeds GPT-4-turbo’s standalone performance in Emails (6.43 vs. 3.79). While TinyLlama-1.1B outperforms GPT-3.5-turbo in Instructions and Paper Abstracts, hybrid systems further amplify performance, particularly in Emails, likely due to the SLM's ability to refine domain-specific nuances. These results validate the hybrid architecture’s ability to balance scalability (LLMs) and efficiency (SLMs) for edge deployments.
The key factor driving this performance is that \model~ utilizes smaller models to fine-tune on local data, employing logit offset to proxy fine-tune inaccessible black-box LLMs. Additionally, \model~ incorporates personalized LoRA modules to address data heterogeneity across different tasks, offering a more tailored approach compared to the generic one-for-all solution. Furthermore, \model~ uses a prompt-wise MoE-based router to effectively handle the properties of diverse real-world requests.

\begin{table}[!ht]
  \centering
  \caption{Performance comparison of LLMs and SLMs.}
  \resizebox{0.95\linewidth}{!}{
    \begin{tabular}{c|ccc}
      \bottomrule[1.5pt]
      \rowcolor{my_c1}
      Model & Instructions & Emails & Paper Abstracts\\
      \toprule[0.75pt]
      \multicolumn{4}{c}{Generic LLM} \\ \hline
      GPT-3.5-turbo & 3.76 & 3.03 & 3.59 \\
      GPT-4-turbo   & 6.46 & 3.79 & \textbf{7.94} \\ \hline
      \multicolumn{4}{c}{Vanilla SLM} \\ \hline
      TinyLlama-1.1B & 7.66 & 5.92 & 7.32 \\ \hline
      \multicolumn{4}{c}{Generic LLM + Specialized SLMs} \\ \hline
      + GPT-3.5-turbo & \textbf{7.84} & 6.23 & 7.65 \\
      + GPT-4-turbo   & 7.82 & \textbf{6.43} & 7.87 \\
      \bottomrule[1.5pt]
    \end{tabular}
  }
  \label{tab:Co}
\end{table}

\begin{figure*}[!ht]
    \centering
    \includegraphics[width=0.85\linewidth]{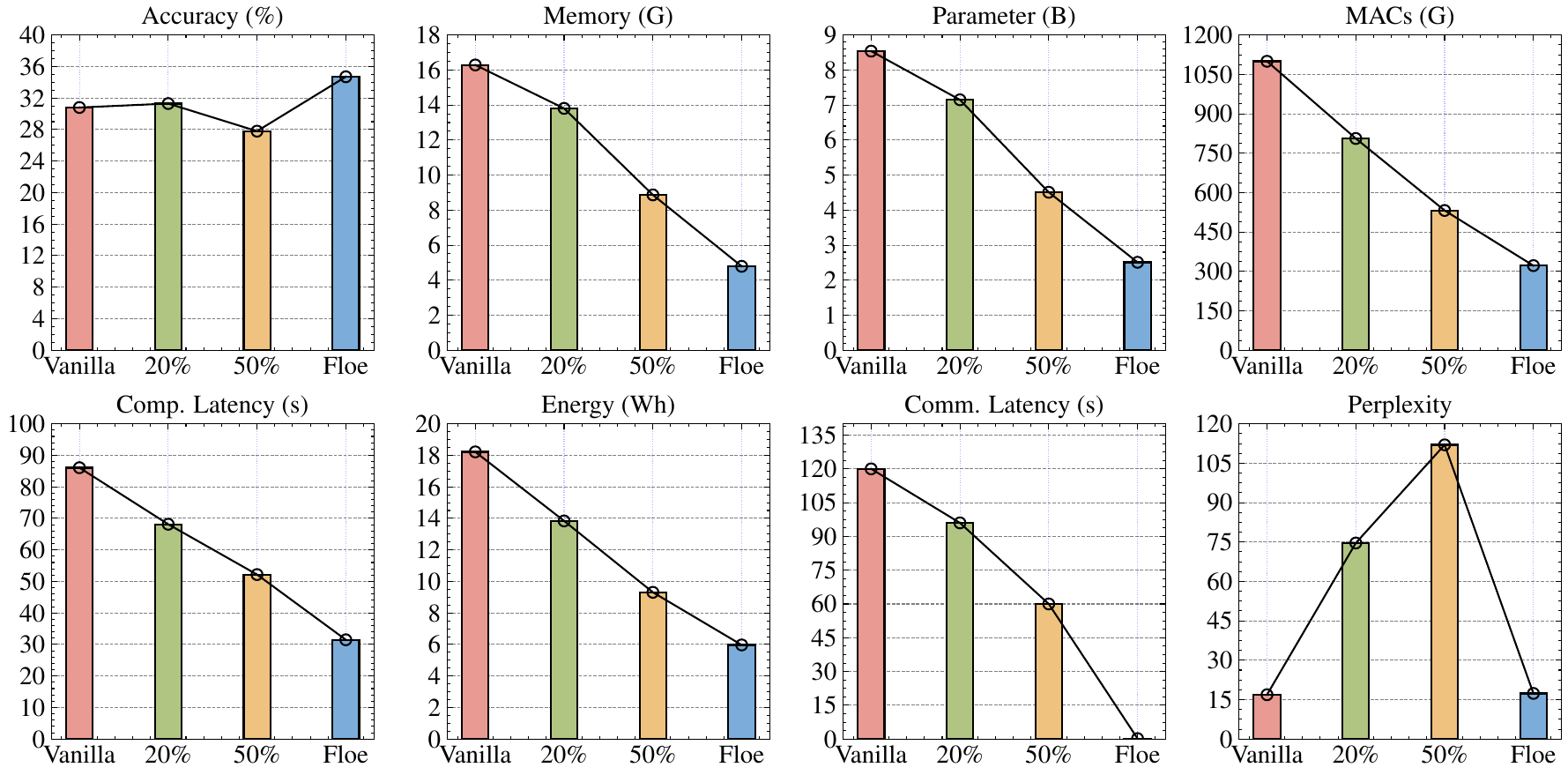}
    \caption{Efficiency comparison of using vanilla model (Gemma-7B), compressed model (rate 20\% and 50\%) and small model (Gemma-2B) as the global collaborative model.}  
    \label{fig:system}
\end{figure*}

\subsection{System Effectiveness}
\noindent \textit{RQ2: How does \model~ enhance system efficiency, particularly in reducing training energy consumption and latency?}

To evaluate the system efficiency of \model~, we examine the detailed statistics of different compression methods, including average model performance on BBH benchmark, perplexity, memory requirements, parameters, MACs (multiply accumulate operations), communication and computation latency, and energy consumption. The experiments are conducted on WikiText-2 and are executed on NVIDIA A40 for vanilla Gemma-7B, and NVIDIA Jetson Nano for other models, monitoring both energy consumption and duration with the Monsoon Power Monitor. 20\% and 50\% denote the model compression ratios; subsequently, the compressed models are deployed as global collaborative models for training. 
As shown in Figure \ref{fig:system}, compared to traditional model compression, \model~ only needs fewer parameters to fine-tune while occupying less memory footprint and obtaining the highest accuracy. 
The underlying driver is that traditional methods broadcast a lossy version as the shared model, yet the required memory remains significantly large. Meanwhile, the small model effectively speeds up the overall runtime, and \model~ achieves significant energy savings. Assuming a network bandwidth of 100MBps, \model~ reduces 99\% upload and download latency, due to its compact size and the efficient use of LoRA modules. This reduction is particularly beneficial for clients with limited network capacity.
It is also extremely user-friendly for resource-constrained devices, particularly those powered by batteries. Notably, while model compression enhances privacy by not exposing full parameters, it still necessitates access to all model parameters during the compression process. 
This highlights \model~ benefits in boosting system efficiency and resource optimization, offering a practical solution for deploying advanced machine learning in resource-constrained environments.

\subsection{ Ablation Study}
 To explore the impact of the distinct modules of \model~, we develop three model variants: (1) $\text{\model~}^{-M}$, where the heterogeneity-aware LoRA trainer is omitted, causing the server to dispatch the same standard model to devices with heterogeneous memory capacities. (2) $\text{\model~}^{-P}$, where the task-specific router/aggregator scheme are disabled. (3) $\text{\model~}^{-R}$, where the prompt-wise parameter-free MoE router is removed to assess its efficacy in handling mixed-task inputs.

\begin{figure*}
    \centering
    \includegraphics[width=0.35\linewidth]{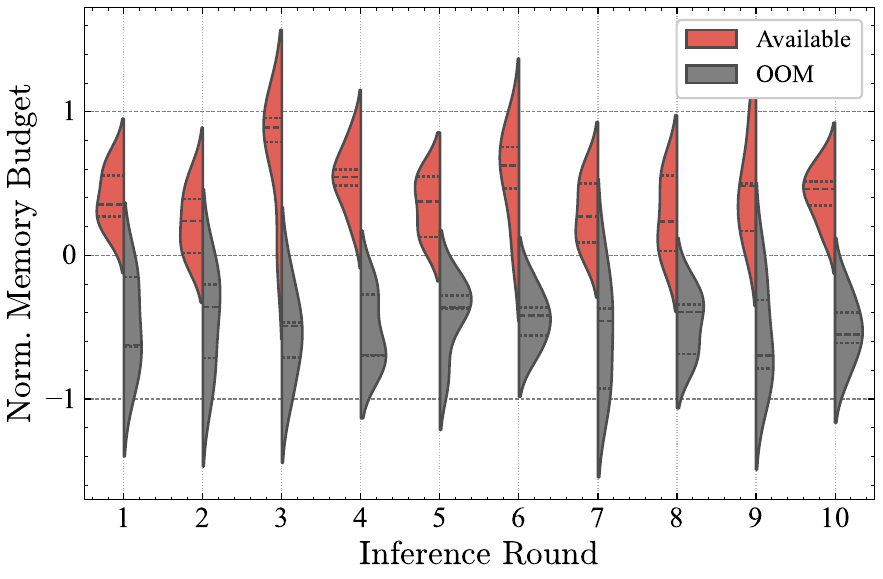}
    \includegraphics[width=0.35\linewidth]{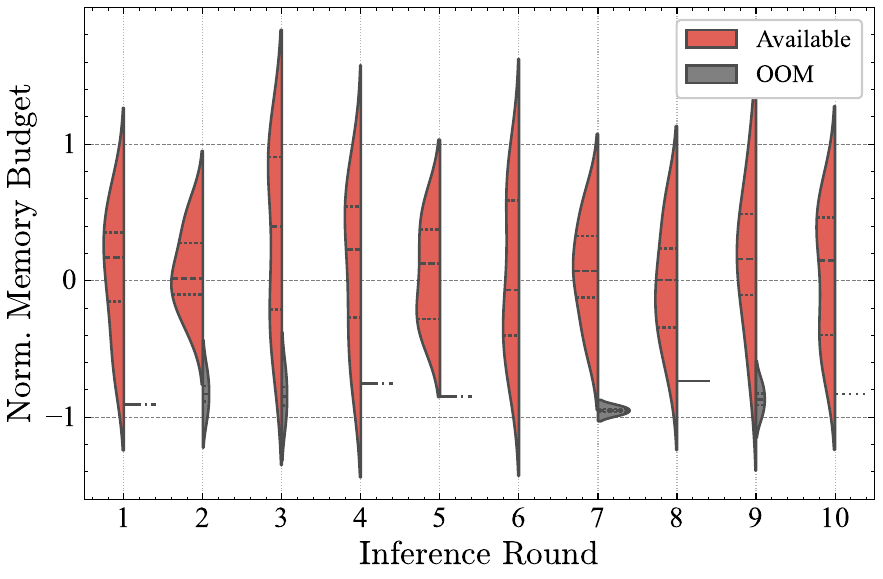}
    \caption{Analysis of model tuning on various edge device memory budgets. Left: $\text{\model~}^{-M}$; Right: \model~. }    
    \label{fig:memory}
\end{figure*}

 \begin{figure*}
    \centering
        \begin{subfigure}{0.15\linewidth}
        \centering
        \includegraphics[width=\linewidth]{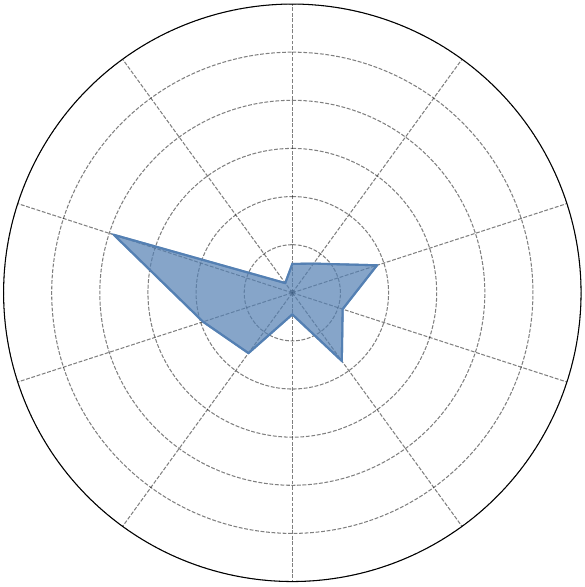}
        \caption{Local-1}
        \label{fig:sub1}
    \end{subfigure}
        \begin{subfigure}{0.15\linewidth}
        \centering
        \includegraphics[width=\linewidth]{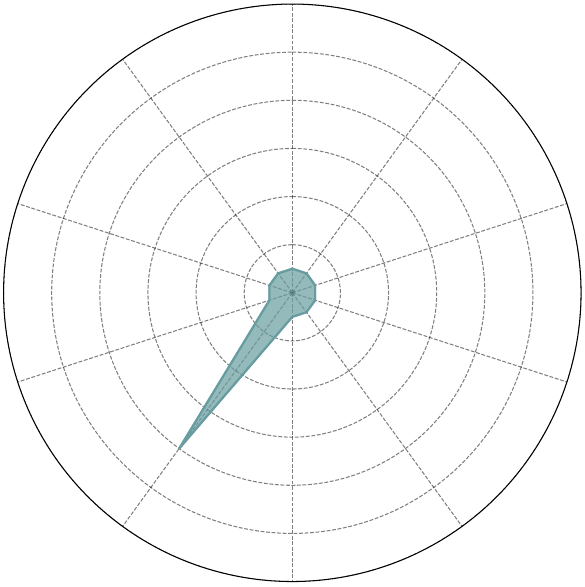}
        \caption{Local-2}
        \label{fig:sub2}
    \end{subfigure}
    \begin{subfigure}{0.15\linewidth}
        \centering
        \includegraphics[width=\linewidth]{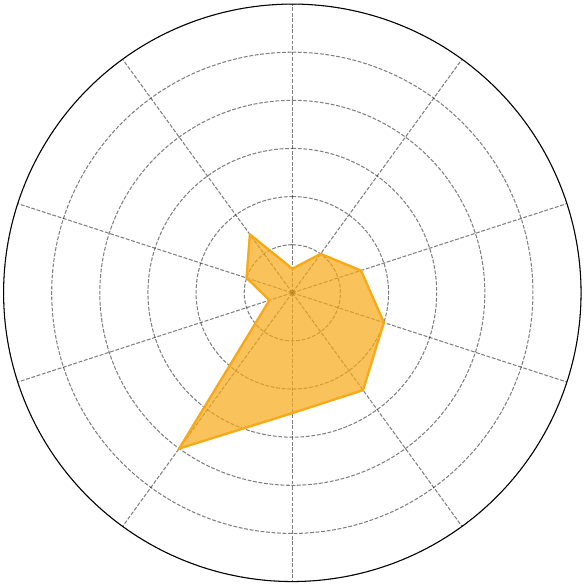}
        \caption{Local-3}
        \label{fig:sub3}
    \end{subfigure}
    \begin{subfigure}{0.15\linewidth}
        \centering
        \includegraphics[width=\linewidth]{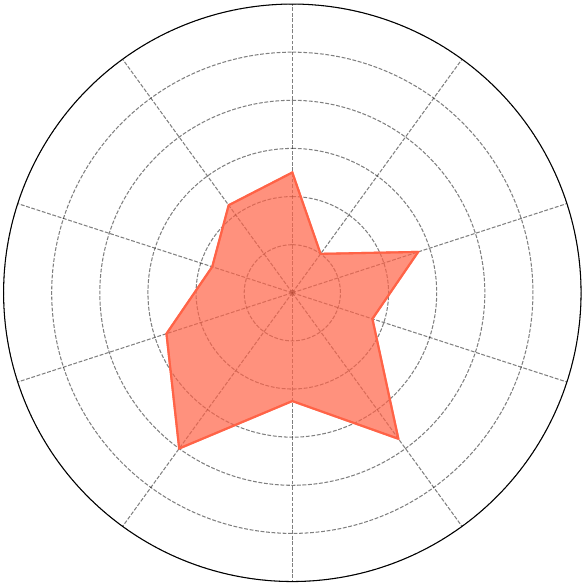}
        \caption{$\model~^{-P}$}
        \label{fig:sub4}
    \end{subfigure}
    \begin{subfigure}{0.15\linewidth}
        \centering
        \includegraphics[width=\linewidth]{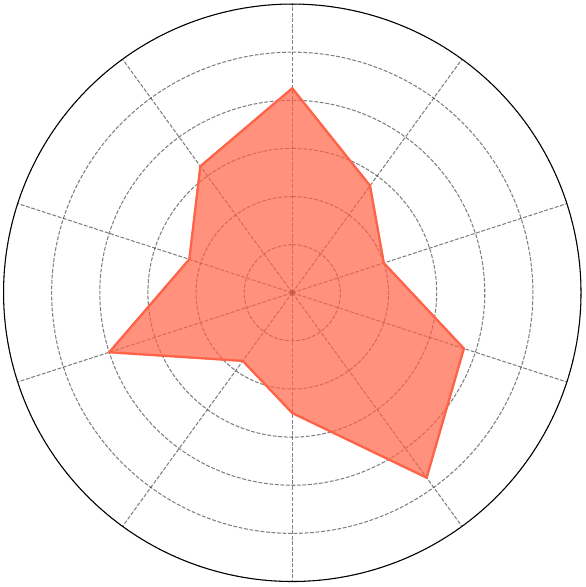}
        \caption{$\model~^{-R}$}
        \label{fig:sub5}
    \end{subfigure}
    \begin{subfigure}{0.15\linewidth}
        \centering
        \includegraphics[width=\linewidth]{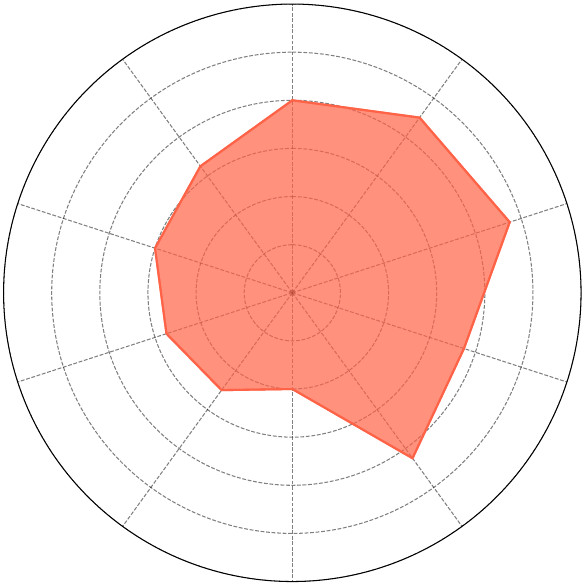}
        \caption{\model~}
        \label{fig:sub6}
    \end{subfigure}
\caption{Performance magnitudes of local fine-tuning, $\model~^{-P}$ (w/o personalized LoRA scheme), $\model~^{-R}$ (w/o MoE router) and \model~ on different downstream tasks, evaluated on BBH benchmark.}
\label{fig:personalized}
\end{figure*}

\noindent \textit{RQ3: What impact does the device-specific model pruner have on the heterogeneity memory budgets?}

Taking a closer look at the strategies employed by the device-specific pruner, Figure \ref{fig:memory} illustrates the model fine-tuning scenarios on hardware-heterogeneous edge devices. As shown in Figure \ref{fig:memory} left, i.e., $\text{\model~}^{-M}$, the memory wall hinders effective model deployment. This limitation restricts the model fine-tuning using the data in resource-constraint devices, consequently diminishing the diversity and quantity of the overall training data and adversely affecting the model performance. Meanwhile, it also undermines the foundational goal of federated learning, which is to facilitate collaboration across multiple devices. As shown in Figure \ref{fig:memory} right, by adaptive rank of the model following the memory budgets of the edge devices, the range of available devices of deployment is expanded. This enhancement boosts the robustness and scalability of \model~, expanding its utility across a broader array of resource-constrained environments.

\noindent \textit{RQ4: What impact do the personalized LoRA scheme and MoE router have on the heterogeneity data and mixed-task inputs?} 

As shown in Figure \ref{fig:personalized}, we present the performance magnitudes of $\model~^{-P}$ and $\model~^{-R}$ across 10 BBH subtasks. We observe that local fine-tuning without federated collaboration excels only in specific tasks, often due to insufficient data. In contrast, $\model~^{-P}$, which fails to address data heterogeneity and merely aggregates updated LoRA modules from local devices, results in suboptimal performance across most tasks. Furthermore, although $\model~^{-R}$ incorporates a personalized LoRA scheme for fine-tuning, it still struggles to handle practical mixed-task inputs effectively. These findings underscore the critical importance of both the personalized LoRA scheme and the MoE router in maintaining $\model~$ performance in environments characterized by diverse data and task demands.

\noindent \textit{RQ5: Can \model~ generalize to open-source and closed-source model pairs?}

To validate \model~'s adaptability beyond the Gemma family, we extended our evaluation to the state-of-the-art Llama-3 ecosystem and the closed-source Gemini-1.5-Flash API. As detailed in Table \ref{tab:ablation_models}, pairing the server-hosted Llama-3.1-70B with edge-optimized Llama-3.2 models reveals that performance gains scale with local capacity: the 3B model achieves a substantially higher boost (+5.6\%) compared to the 1B variant (+3.3\%), suggesting that stronger local experts can more effectively leverage global logits. Furthermore, the framework demonstrates robust compatibility with external APIs, delivering consistent improvements (+2.6\%) even when utilizing the lightweight Gemini-1.5-Flash backend, thereby confirming \model~'s effectiveness across diverse open-source and proprietary architectures.

\begin{table}[!t]
\centering
\caption{Ablation Study on Model Heterogeneity.}
\label{tab:ablation_models}
\resizebox{0.95\linewidth}{!}{
\begin{tabular}{llcc}
\bottomrule[1.5pt]
\rowcolor{my_c1} \textbf{Server / Edge Pair} & \textbf{Method} & \textbf{BBH Acc.} & \textbf{Gain} \\ \hline
\multirow{3}{*}{\begin{tabular}[c]{@{}l@{}}Server: Llama-3-70B\\ Edge: Llama-3-1B\end{tabular}} 
 & SLM-Only & 32.5\% & - \\
 & LLM-Only & 55.9\% & - \\
 & \model~ & 59.2\% & +3.3\% \\ \hline
\multirow{2}{*}{\begin{tabular}[c]{@{}l@{}}Server: Llama-3-70B\\ Edge: Llama-3-3B\end{tabular}} 
 & SLM-Only & 44.1\% & - \\
 & \model~ & 61.5\% & +5.6\% \\ \hline
\multirow{2}{*}{\begin{tabular}[c]{@{}l@{}}Server: Gemini-1.5-Flash\\ Edge: Llama-3-1B\end{tabular}} 
 & LLM-Only (API) & 51.2\% & - \\
 & \model~ & 53.8\% & +2.6\% \\ \bottomrule[1.5pt]
\end{tabular}}
\end{table}

\subsection{Hyperparameter Sensitivity}
\noindent\textit{Number of experts.} To thoroughly investigate the relationship between model performance and the number of LoRA modules, we systematically vary the quantity of LoRA components during proxy inference and evaluate their downstream task performance on the challenging BBH benchmark. As illustrated in Figure \ref{fig:param}, increasing the number of LoRA modules generally correlates with measurable improvements in task outcomes, demonstrating a clear trend toward enhanced model capability. This phenomenon can be explained by the increased architectural flexibility afforded by a larger pool of LoRA modules, which creates an expanded configuration space for dynamic parameter adaptation. Specifically, the greater diversity of low-rank matrices enables more precise specialization of neural network pathways in response to the heterogeneous input distributions characteristic of mixed-task scenarios.
The observed performance gains stem from two synergistic mechanisms: (1) The expanded module pool allows the router network to select task-optimal combinations of LoRA components, effectively creating a hierarchical adaptation process where different sub-modules specialize in distinct input patterns or knowledge domains. (2) The increased configuration granularity prevents the parameter interference issues common in monolithic adaptation approaches, as specialized LoRA modules can independently adjust their rank-decomposed weightings without disrupting other functional components. This modular architecture not only enhances adaptability to complex input structures but also maintains parameter efficiency by focusing computational resources on task-critical pathways.

\noindent\textit{Routing alignment accuracy.} To explicitly evaluate the decision-making quality of the MoE-Router (beyond downstream perplexity), we measured its routing alignment accuracy. We mapped the BBH tasks into the discovered latent cluster domains (e.g., Logical Reasoning, Knowledge, Linguistics) and tracked the Top-1 expert selection frequency. The router achieves an average alignment accuracy of 87.6\% across all tasks. Specifically, for distinct domains like ``Boolean Expressions'' and ``Object Counting,'' the routing accuracy exceeds 92\%, indicating high stability. Misclassifications primarily occur in ambiguous tasks (e.g., ``Date Understanding''), yet the soft-gating mechanism mitigates this by blending experts, preventing catastrophic performance drops.

\noindent\textit{Dynamic reweighting.} To decode the decision logic of the alignment network, Fig. \ref{fig:weight_dist} visualizes the distribution of fusion weights $w$ across the BBH benchmark. The weights span the full $[0, 1]$ spectrum with a mean of $\approx 0.725$, confirming dynamic arbitration rather than a static average. This distinct skew towards $w > 0.5$ empirically validates the \textit{Specialization Hypothesis}: the fine-tuned SLM typically yields lower entropy (higher confidence) than the zero-shot Cloud LLM. Consequently, the network prioritizes the local expert for domain-specific tokens, while reserving the Cloud LLM ($w < 0.5$) as a fallback regularizer during high-uncertainty scenarios.

\begin{figure}[!t]
  \centering
  \includegraphics[width=0.95\linewidth]{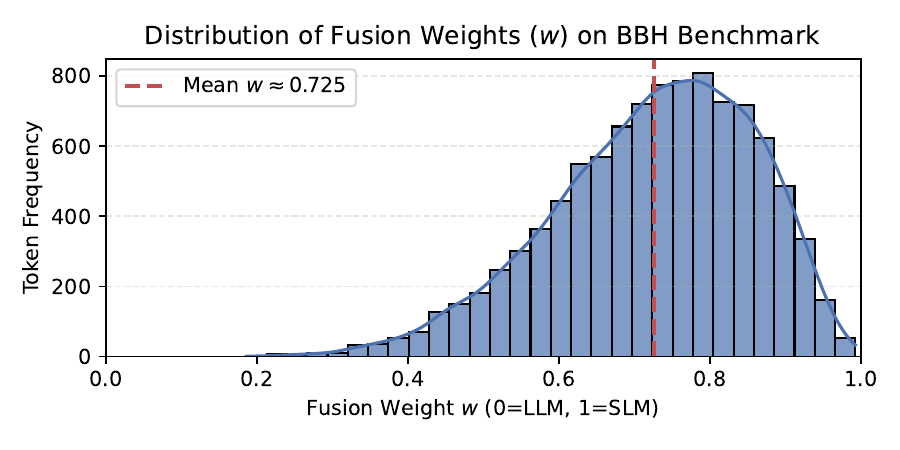}
  \caption{Distribution of dynamic fusion weights.}
  \label{fig:weight_dist}
\end{figure}

\begin{figure}[!t]
    \centering
    \includegraphics[width =0.95\linewidth]{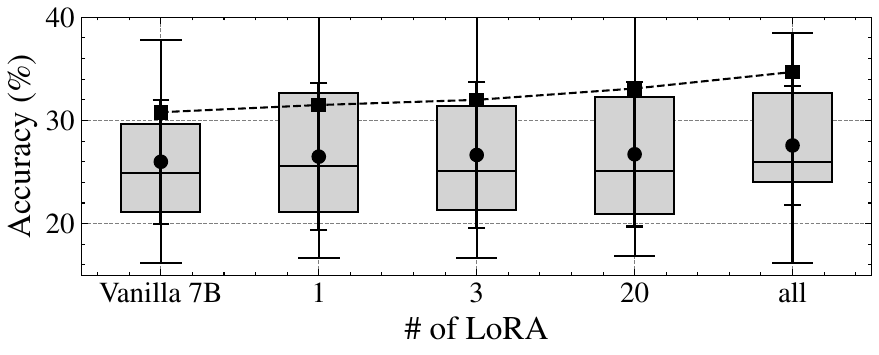}
\caption{Hyperparameter sensitivity of the number of experts.}
    \label{fig:param}
\end{figure}

\begin{figure}[!t]
    \centering
    \includegraphics[width=0.85\linewidth]{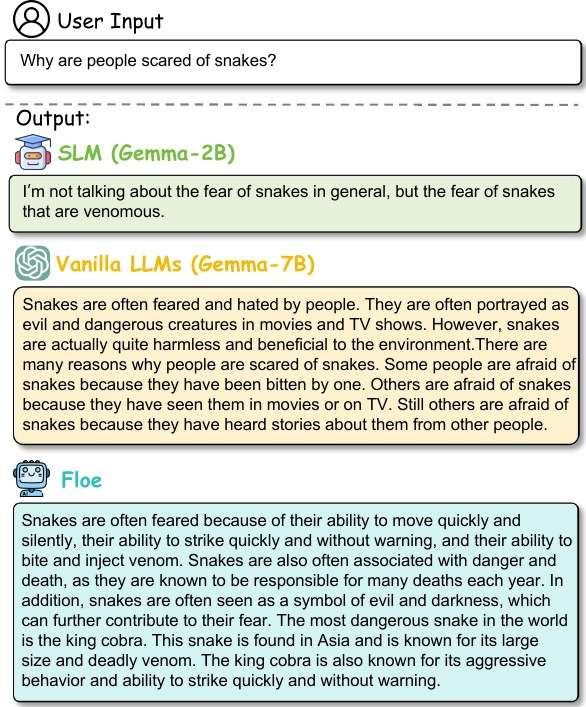}
    \caption{Case study. With identical inputs, \model~ generates more detailed and comprehensive outputs than SLM and LLM.}
    \label{fig:case_study}
\end{figure}

\noindent\textit{Impact of Network Latency}
To empirically validate the robustness of token-level communication, we evaluated the end-to-end generation latency on a Jetson Orin NX device while simulating varying network Round-Trip Times (RTT) ranging from 0ms to 500ms. As illustrated in Figure \ref{fig:Latency_Analysis}, 
for masked region (RTT $< 100$ms), the total latency curve remains flat and identical to the standalone SLM baseline ($\approx$ 65ms/token). In this region, the communication overhead is completely masked by the local computation time ($t_{net} < t_{edge} - t_{cloud}$), resulting in zero effective latency penalty. For bounded region (RTT $> 100$ms), as network conditions degrade, the latency momentarily rises but is immediately capped by the fallback timeout threshold ($\mathcal{T}_{wait}$). The plateau in Figure R1 confirms that even under catastrophic network lag, the end-to-end latency does not explode but gracefully reverts to the local SLM speed. These results demonstrate that Floe's real-time performance is robust: it benefits from cloud intelligence when the network permits, while strictly preventing communication overhead from dominating the inference cycle in adverse environments.

\begin{figure}[!ht]
    \centering
    \includegraphics[width=0.85\linewidth]{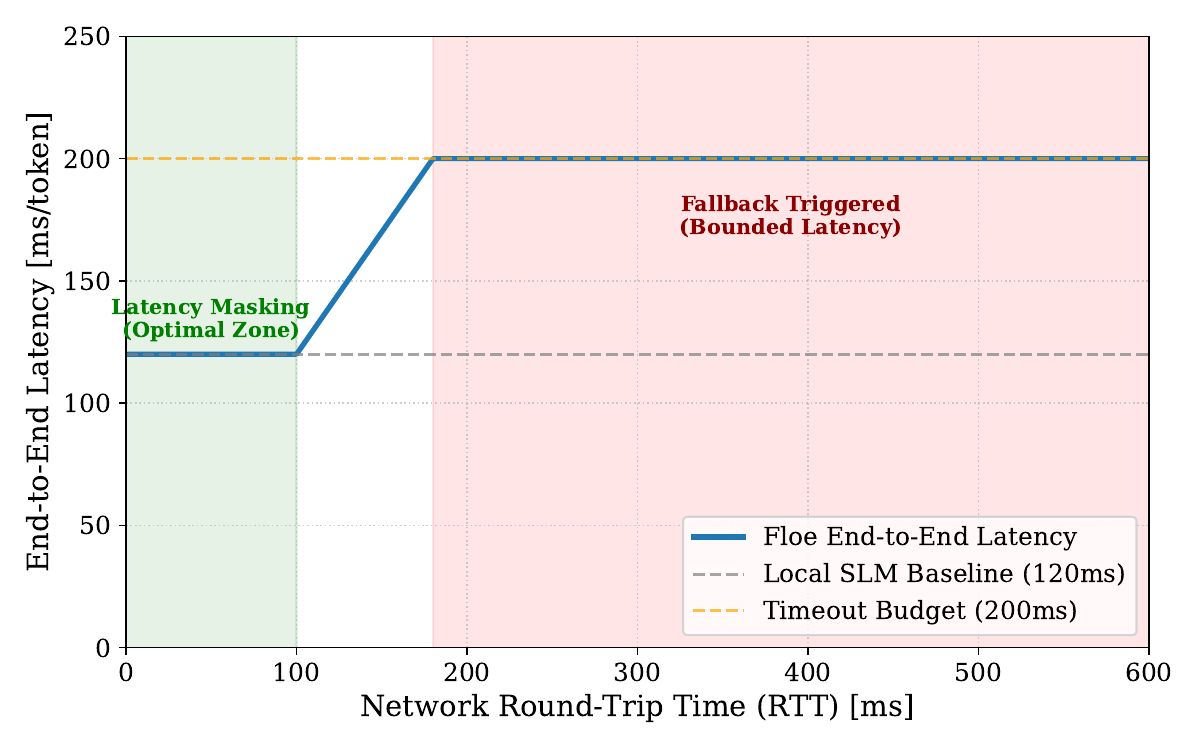}
    \caption{Impact of network RTT on end-to-End latency. The communication overhead is either masked by parallel local computation (Green Zone) or strictly bounded by the fallback timeout (Red Zone), ensuring stable real-time response.}
    \label{fig:Latency_Analysis}
\end{figure}

\subsection{Privacy Evaluation}
We evaluate the effectiveness of the two-stage privacy detector described above using a labeled subset of the CoGenesis dataset, consisting of 3,000 real-world prompts annotated as either sensitive (e.g., containing personal-related content) or non-sensitive. The detector achieves an F1 score of 94.3\%, with 97.1\% precision and 91.7\% recall, demonstrating strong accuracy in privacy classification.
Importantly, over 93\% of prompts identified as sensitive are correctly routed to the on-device SLM, fully bypassing the cloud LLM and preventing any personal data from leaving the edge. These results validate the design of the detector: the rule-based Stage 1 offers low-latency coverage of explicit entities (e.g., phone numbers, medical terms), while the semantic Stage 2 reliably captures paraphrased or implicit privacy cues via cosine similarity to category centroids. This combination ensures both efficiency and robustness in privacy-sensitive prompt routing.

\subsection{Case Study}
\noindent \textit{\textbf{RQ7:} How does \model~ perform in practice inputs compared to the vanilla models?}

We present a comparative analysis of output samples generated \model~. As demonstrated in the case studies of Figure \ref{fig:case_study}, \model~ consistently produces responses that surpass baseline models across multiple quality dimensions. Specifically, the generated text exhibits enhanced fluency through more natural syntactic structures, improved relevance via precise contextual alignment with input prompts, and increased informativeness characterized by richer detail density and conceptual depth. This performance advantage stems from our novel methodology that synergistically combines two key innovations: (1) Proxy fine-tuning enables effective knowledge distillation by leveraging federated learning principles, where the small-scale model acts as a local knowledge integrator that captures task-specific patterns while preserving privacy constraints. This process creates a robust initialization point for subsequent adaptation. (2) The prompt-wise MoE routing mechanism implements dynamic architecture selection, wherein a trainable router network adaptively activates the most appropriate LoRA modules based on input characteristics. This granular parameter-efficient approach allows precise control over knowledge fusion from different expert modules, resulting in optimized output quality without compromising computational efficiency. The combination of these techniques creates a hierarchical knowledge transfer framework that enhances both the semantic coherence and functional adaptability.

%% file: sec/07_con.tex
\section{Limitations}
Despite its promising results, \model~ still faces several important limitations. (1) The two-stage Privacy Detector is a heuristic, not a formal privacy guarantee. As identified in our evaluation (Section V.F), the detector achieves 91.7\% recall, which implies that a percentage of sensitive queries may be misclassified as ``General Task Requests''. In such cases, this sensitive information would be unintentionally leaked to the cloud-hosted LLM. There is an inherent trade-off in tuning the detector's similarity threshold ($\tau$): a more conservative threshold would increase recall (improving privacy) but reduce precision, causing more non-sensitive prompts to be unnecessarily restricted to the on-device SLM. (2) The logit-level fusion mechanism (Section IV.C)  fundamentally assumes that the vocabulary distributions of the cloud LLM and the on-device SLM are known and mappable. While this is feasible for models with open-source tokenizers like GPT-4 (using tiktoken) or for fully open-source stacks (like our Gemma-7B/2B experiment ), this approach would fail if a black-box LLM provider does not expose its tokenizer or logit vocabulary. In such a ``fully black-box'' scenario, logit fusion is infeasible. (3) The heterogeneity-aware LoRA trainer (Section III.B) dynamically selects the rank $r_i$ based purely on system constraints, namely the device's memory budget ($\mathcal{B}_i$) and a global latency bound ($\mathcal{T}$). This approach prioritizes system feasibility and maximizes device participation by preventing OOM errors. However, this creates a trade-off: a client with highly valuable or abundant local data but weak hardware will be assigned a low rank. This low rank limits the capacity of its update, potentially preventing the final aggregated model from fully benefiting from that client's rich data.

\section{Conclusion}
This paper introduces \model~, a high-performance federated learning framework that synergizes cloud-hosted black-box LLMs with edge-deployed lightweight SLMs to address privacy preservation, resource constraints, and model ownership challenges in decentralized LLM fine-tuning and inference. The framework ensures data privacy and minimizes memory overhead by retaining sensitive user data and personalization on edge devices, while leveraging the LLM’s general knowledge through a novel logits-based coordination mechanism. To accommodate edge hardware diversity, \model~ employs a heterogeneity-aware LoRA fine-tuning strategy that dynamically adapts model parameters to varying resource budgets. Additionally, a parameter-efficient prompt-wise MoE router enhances adaptability to diverse input distributions.
Unlike prior approaches that focus solely on accuracy, \model~ offers a deployment-friendly and scalable solution, making it well-suited for real-time, privacy-critical applications on resource-constrained devices.

\section*{Acknowledgments}
This research received support from the Science and Technology Development Fund of Macau (0107/2024/RIA2), Joint Science and Technology Research Project with Hong Kong and Macau in Key Areas of Nansha District's Science and Technology Plan (EF2024-00180-IOTSC) and the Multi-Year Research Grant of University of Macau (MYRG-GRG2023-00211-IOTSC-UMDF, MYRG-GRG2024-00180-IOTSC). Please ask Dr. Li Li (llili@um.edu.mo) for correspondence.